\documentclass[a4paper,11pt]{article}
\usepackage[english]{babel}

\usepackage[latin1]{inputenc}
\usepackage{amstext, amssymb, amsthm, amsfonts, latexsym,bbm}
\usepackage[tbtags]{amsmath}
\usepackage{mathtools}
\usepackage{hyperref}
\usepackage{color}
\usepackage{graphicx}
\usepackage{enumerate}
\usepackage{caption, subcaption}
\usepackage{authblk}
\usepackage{multirow}
\usepackage{rotating}
\usepackage{tocloft}

\usepackage[top=1in, bottom=1.25in, left=1in, right=1.5in]{geometry} % for the margin

\DeclareMathOperator{\Var}{Var}

\DeclareMathOperator*{\sign}{sign}

\newcommand{\R}{\mathbb{R}}

\newcommand{\Prob}{\mathbb{P}}
\newcommand{\E}{\mathbb{E}}

\renewcommand{\geq}{\geqslant}
\renewcommand{\leq}{\leqslant}

\newcommand{\lp}{\left(}
\newcommand{\rp}{\right)}
\newcommand{\lc}{\left[}
\newcommand{\rc}{\right]}

\newtheorem{thm}{Theorem}[section]
\newtheorem{prop}[thm]{Proposition}

\newtheorem{lemma}[thm]{Lemma}
\newtheorem{corollary}[thm]{Corollary}

\theoremstyle{definition}

\newtheorem{assump}[thm]{Assumption}

\begin{document}

% Include your paper's title here

\title{The feasibility of equilibria in large ecosystems: \linebreak a primary but neglected concept in the complexity-stability debate}

\author[1]{Michael Dougoud}
\author[1,3]{Laura Vinckenbosch}
\author[2]{Rudolf Rohr}
\author[2]{Louis-F\'elix Bersier}
\author[1,*]{Christian Mazza}

\affil[1]{University of Fribourg, Department of Mathematics, Chemin du Mus\'ee 23, CH-1700 Fribourg}
\affil[2]{University of Fribourg, Department of Biology, Unit of Ecology and Evolution, Chemin du Mus\'ee 10, CH-1700 Fribourg}
\affil[3]{University of Applied Sciences Western Switzerland - HES-SO, HEIG-VD, CH-1401 Yverdon-les-Bains}
\affil[*]{Corresponding author : christian.mazza@unifr.ch}

\date{\vspace{-1ex}\today\vspace{-1ex}}
\maketitle

% figures
%\graphicspath{{./fig//}}

% ABSTRACT

\begin{abstract}
The consensus that complexity begets stability in ecosystems was challenged in the seventies, a result recently extended to ecologically-inspired networks. The approaches assume the existence of a feasible equilibrium, i.e. with positive abundances. However, this key assumption has not been tested. We provide analytical results complemented by simulations which show that equilibrium feasibility vanishes in species rich systems. This result leaves us in the uncomfortable situation in which the existence of a feasible equilibrium assumed in local stability criteria is far from granted. We extend our analyses by changing interaction structure and intensity, and find that feasibility and stability is warranted irrespective of species richness with weak interactions. Interestingly, we find that the dynamical behaviour of ecologically inspired architectures is very different and richer than that of unstructured systems. Our results suggest that a general understanding of ecosystem dynamics requires focusing on the interplay between interaction strength and network architecture.
\end{abstract}

\setcounter{tocdepth}{2}

\tableofcontents

%%%% TEXT

\section{Introduction}

A central question in ecology is to understand the factors and conditions that ensure ecological systems to persist, a requisite for the sustained provisioning of vital ecosystem services. This question of a ``balance of nature" has a long history in science~\cite{Kingsland1985,Cuddington2001}, and the consensus that ``complexity begets stability" emerged among ecologists in the fifties. MacArthur~\cite{McArthur55} had a radical view on this question, arguing that stability will increase with the two fundamental ingredients of complexity, the number of species and of interactions. The argument for this claim was borrowed from Odum~\cite{Odum1953}: stability increases with the number of paths through which energy can flow up in a food web. Later, Elton~\cite{Elton58} provided a suite of arguments for this positive relationship. The first one is the following: mathematical systems composed of one predator and one prey exhibit conspicuous fluctuations. Implicit in this argument is that more complex systems should be more stable, which remained untested at that time. In the seventies, Levins~\cite{Levins1968, Levins1974}, Ashby and Gardner~\cite{Gardner1970}, and May~\cite{May1971} showed numerically that large random systems may be expected to be stable up to a certain connectance threshold, contradicting the earlier ideas that complex natural systems are more likely to be stable. In his impactful work, May~\cite{May1972, May2001Book} showed mathematically using random matrix theory~\cite{Wigner1958,Ginibre1965} that large and random ecosystems are inherently unstable. His approach was based on a mathematical study of community matrices, which represented unstructured random networks of interacting species. He used a local-stability analysis assuming these systems were at equilibrium. He derived a simple and elegant criterion for system stability, which is a milestone in the stability-complexity debate~\cite{Bersier2007}. May~\cite{May2001Book} concluded that there was no comfortable theorem assuring that increased complexity will lead to stable systems, and that the task was therefore to ``elucidate the devious strategies which make for the stability in enduring natural systems".

Recently, the work of May was revisited by Allesina and collaborators~\cite{Allesina2012,Allesina2015}. They established stability criteria for systems where species interact specifically via either competition, mutualism, or predation. As such, their contribution is a refinement of May's approach that considered mixtures of interaction types. The general conclusion is that, in species-rich communities, interactions should be moderate to ensure the stability of the equilibrium. They also performed simulations to study stability in randomly generated systems whose architecture mimics empirically observed food webs~\cite{Cohen1990community,Williams2000simple}. They arrived at the counter-intuitive conclusion that such ecologically-inspired structured graphs are less stable than unstructured ones. They also focused on interaction weights, and interestingly found that weak interactions should increase the stability for mutualistic and competitive webs, but decrease the stability of food webs (see also~\cite{Gellner2016}, where more precise statements are obtained for interactions weights of different intensity and symmetry). It was also found that more realistic structures seem to be detrimental for stability, and that the structure alone plays a minor role for stability~\cite{James2015,Johnson2014}.

All the above approaches are local-stability analyses, where systems are linearised at the equilibrium point and stability is evaluated only in the close vicinity of this point. In ecological systems, this point is meaningful only if the equilibrium-abundances are all strictly positive; in other words, if the equilibrium is feasible. In these previous works, equilibria were simply assumed feasible, without further analysis; basically, the true Jacobian, evaluated at the equilibrium, was replaced by a random matrix. Interestingly, soon after the work of May~\cite{May2001Book}, Roberts~\cite{Roberts1974} noted that May's approach indeed remained silent concerning the feasibility of the equilibria. In a simulation study, he found that, in feasible systems, stability increased with the number of species. However, Roberts did not explore the very question of the relationship between system's size and feasibility probability. This is a key issue since feasibility is a prerequisite to local-stability analysis. In the particular case of competitive system, this question had to await the work of Logofet~\cite{Logofet1993Book} who found that ``equilibriumness" was vanishing with species richness. Since then, few contributions have explored the question of feasibility, and if so in particular systems and with regard to the characteristics of the interactions rather than to species number (e.g., for asymmetric competition~\cite{Bastolla2005,Nattrass2012} or for mutualism~\cite{Rohr2013}). Note that already in 1970, Vandermeer~\cite{vandermeer1970} studied a question very related to feasibility, the expected number of species that can coexist in competitive communities.

Here, we present an extensive study of the fundamental question of equilibrium feasibility, and expose its underlying mathematical principles. Two classes of networks are studied: 1) networks which do not possess any particular topological structure, for which we consider random, mutualistic, competitive, and predator-prey interactions, and 2) predator-prey networks designed to capture food web architecture, following the cascade~\cite{Cohen1990community}, niche~\cite{Williams2000simple} and nested-hierarchy~\cite{Cattin2004} models. We show that, in situations compatible with the above works~\cite{May1972,Allesina2012,Allesina2015}, the existence of a feasible equilibrium is not guaranteed, since the probability to observe such a point decreases exponentially towards zero as the size of the network grows. Interestingly, we find that with weak interactions, non-trivial feasible equilibria are found only in ecologically-inspired systems, in a way that is ecologically sensible.

%\textcolor{violet}{Most of the studies based on random webs focus on random community matrices and study stability by checking whether all eigenvalues have negative real parts. For Lotka-Volterra dynamics, which form the most used family of dynamical models, the community matrix is given by the Jacobian $J(x)$ of the related vector field evaluated at the equilibrium or steady state $x^*$. These studies neglect  $x^*$ by replacing the true Jacobian by a random matrix. 
%The very question of the existence and the feasibility (positivity) of such an equilibrium, be it stable or not, has remained neglected, with some exceptions tackling particular situations\cite{vandermeer1970,Logofet1983Book,Nattrass2012,Rohr2013}.  }

\section{Results}

We adopt the following strategy: for models of interacting species, we use the demographic parameters most favourable to feasibility, and randomise the interactions to estimate the probability to find equilibria with only positive abundances. This is achieved for different species richnesses, network architectures, and interaction intensities. Additionally, in cases where the equilibrium is trivially feasible, we extend our analyses by investigating the effect of the demographic parameters on feasibility. We provide analytical results for cases compatible with May's approach, and otherwise rely on simulations.

We consider a classical Lotka-Volterra model~\cite{Gillman1997} for large, complex ecological networks with $S$ species. The per capita effect of species $j$ on species $i$ is encapsulated in a coefficient $a_{ij}$ ($a_{ij}<0$ and $a_{ji}<0$ for competition; $a_{ij}>0$ and $a_{ji}>0$ for mutualism; $a_{ij}>0$ and $a_{ji}<0$ for predation of $i$ on $j$). If $a_{ij}=0$, species $j$ has no direct effect on $i$,  which happens with a probability $1-C$, where $C$ denotes the connectance of the graph ($C=L/(S\cdot(S-1))$ with $L$ the total number of links in the network). The  interactions $a_{ij}$ form the $S \times S$ interaction matrix $A$. 

In order to exhibit the effects of network structure on the dynamics, the $a_{ij}$s are randomised with fixed mean and standard deviation. Their magnitude plays a key role for the local stability of networks~\cite{McCann1998,Emmerson2004}. %Therefore, we introduce a normalisation constant on the interaction strengths so that the $a_{ij}$s are divided by the linkage density (a measure of complexity given by $CS$\cite{May1972}) raised to the power $\delta\geq0$ (see ref.\cite{Stouffer2005,Rohr2013}).
Therefore, we introduce a parameter $0\leq\delta\leq1$ ruling the average interaction spread. The $a_{ij}$s are divided by the linkage density (a measure of complexity given by $CS$~\cite{May1972}) raised to the power $\delta$ (see ref.~\cite{Stouffer2005,Rohr2013}). This normalisation is mathematically sensible since Wigner's theory~\cite{Wigner1958}, on which May's and Allesina's results are based, is built for the case $\delta=1/2$. Analytically, we will show that three regimes emerge when $S$ becomes large: strong ($0\leq\delta<1/2$), moderate ($\delta=1/2$), and weak interactions ($1/2<\delta\leq1$). Intraspecific competition is included as customarily~\cite{May1972,May2001Book,Allesina2012} with a common coefficient $\theta<0$, separated from the interaction matrix and unaffected by the normalisation constant. Mathematically, in this setting, the vector of abundances $x$ solves the following system of differential equations 
$$\frac{{\rm d}x_i}{{\rm d}t} =  x_i \lp r_i+\theta x_i+ \frac{1}{(CS)^{\delta}} \sum_{j=1}^S a_{ij}\,x_j\rp, \quad\mbox{ for all }  1 \leq i \leq S$$
with $r_i$ the intrinsic growth rates of species $i$. In matrix notation, it becomes
$$\dot x=x\circ(r+(\theta I + (CS)^{-\delta} A)x )$$
with $I$ the $S\times S$ identity matrix, and where $\circ$ denotes the product defined by $x\circ y=(x_1y_1,\ldots,x_Sy_S)$.

%Analytically, three regimes emerge when $S$ becomes large: strong interactions ($0\leq\delta<1/2$), which lead to equilibria having infinite standard deviations; moderate interactions ($\delta=1/2$), which lead to well-defined random equilibria 
% \textcolor{red}{(i.e.~with finite non-zero standard deviations)}
% ; weak interactions ($1/2<\delta\leq1$) for which equilibria are asymptotically deterministic 
% \textcolor{red}{(i.e.~with zero standard deviations)}, with the notable exception of ecologically-inspired networks. 

\subsection{Existence of a feasible equilibrium} 
A system is feasible if all abundances at equilibrium are positive. If this equilibrium exists, the vector $x^*$ corresponds to a point for which the dynamics of the Lotka-Volterra system stops varying, and is given by
\begin{equation}\label{equilibria}
x^*=-(\theta I + (CS)^{-\delta} A )^{-1}r \,.
\end{equation}
If the solution to this equation involves negative abundances, then the system is not feasible. Since the interactions $a_{ij}$ are random variables, it is natural to study $P_S$, the probability to observe a feasible equilibrium, which is given by
\begin{equation}\label{ProbaFeas}
P_S=P(x_i^* > 0,\ \hbox{for all species}\ i).
\end{equation} This probability depends on  $\theta$, $C$, $S$, $\delta$, $A$ and $r$. 
For a given $\delta$, $C$, and type of network,
we choose the growth rates vector $r=(r_i)$ most favourable for feasibility, or, in case of trivially feasible equilibria, $r_i$ to be independent and bounded random variables. The most favourable growth rates correspond to the deterministic vector making the equilibrium to be, on average, as far as possible from the boundaries of the feasibility domain. This deterministic vector is referred to as the mean structural vector. By definition, the feasibility domain is the positive orthant of the phase space. We show  in the Appendix that this vector is related to a particular way of choosing intrinsic growth rates, which have been used in ref~\cite{James2015} to avoid  negative abundances.
  
It is now possible to compute $P_S$ with regard  to the randomness of the $a_{ij}$s.
Extending results on large systems of random equations~\cite{Geman1982}, we develop analytical formulas for $P_{S}$ for our first class of unstructured models when $S$ tends towards infinity, and use Monte Carlo simulations to estimate $P_S$ for the second class of structured models (see Table 1). Before exposing our results, we briefly discuss the link between stability and feasibility.

%%%%%%%%%%%%%%%%%%%%%%%%%%%%%%%%%%%%%%%%%%%%%%%%%%

\subsection{Link between stability and feasibility in random models} 
%\textcolor{red}{%To be consistent with May\cite{May1972}, we 
%Assume for now that the elements of $A$  are such that all diagonal entries are set to 0, and  that all off-diagonal random elements are independent and  set to 0 with probability $1-C$. A non-zero entry is sampled from any distribution with standard deviation $\sigma$.} 
It is only when $A$ is such that $x^*$ is feasible that the local-stability analysis of $x^*$ is sensible. This analysis involves the Jacobian matrix evaluated at $x^*$ (the so-called community matrix), 
\begin{equation}\label{Jacobian}
J(x^*)=\textrm{diag}(x^*)\lp \theta I +(CS)^{-\delta}  A\rp,
\end{equation}
which of course depends on the random vector $x^*$ and on the random interactions matrix $A$. The equilibrium $x^*$ is linearly stable when  the real parts of all eigenvalues of $J(x^*)$ are negative.
Under May's approach, and similarly in ref.~\cite{Allesina2012,Allesina2015,Gellner2016}, $J(x^*)$
is replaced by the random matrix  
$\tilde J = \theta I +(CS)^{-\delta}  A$.
All the information on $x^*$ and on the inherent relations between $x^*$ and $A$ (see equation (1)) are overlooked. The problematic point is that one obtains stability results by focusing on the eigenvalues of
 $\tilde J$ even when a feasible equilibrium $x^*$ does not exist.
 
For example, May~\cite{May1972} considers that the elements of $A$  are such that all diagonal entries are set to 0, and  that all off-diagonal random elements are independent and  set to 0 with probability $1-C$. A non-zero entry is sampled from any centred distribution with standard deviation $\sigma$.
 Applying results from random matrix theory~\cite{Ginibre1965}, the eigenvalues of $\tilde J$ will then have negative real parts for large $S$ when
\begin{equation}\label{May}
\frac{\sigma}{(CS)^\delta}\sqrt{CS} < \vert\theta\vert,
\end{equation}
which is May's stability condition. We show in the Appendix that May's criterion still holds for the matrix $J(x^*)$ under the additional assumption that $x^*$ is feasible (see Fig.~\ref{fig:main1} (b)). Thus it appears critical to study the feasibility of such systems.

% Fig. stability and feasibility
\begin{figure}[H!]
\includegraphics{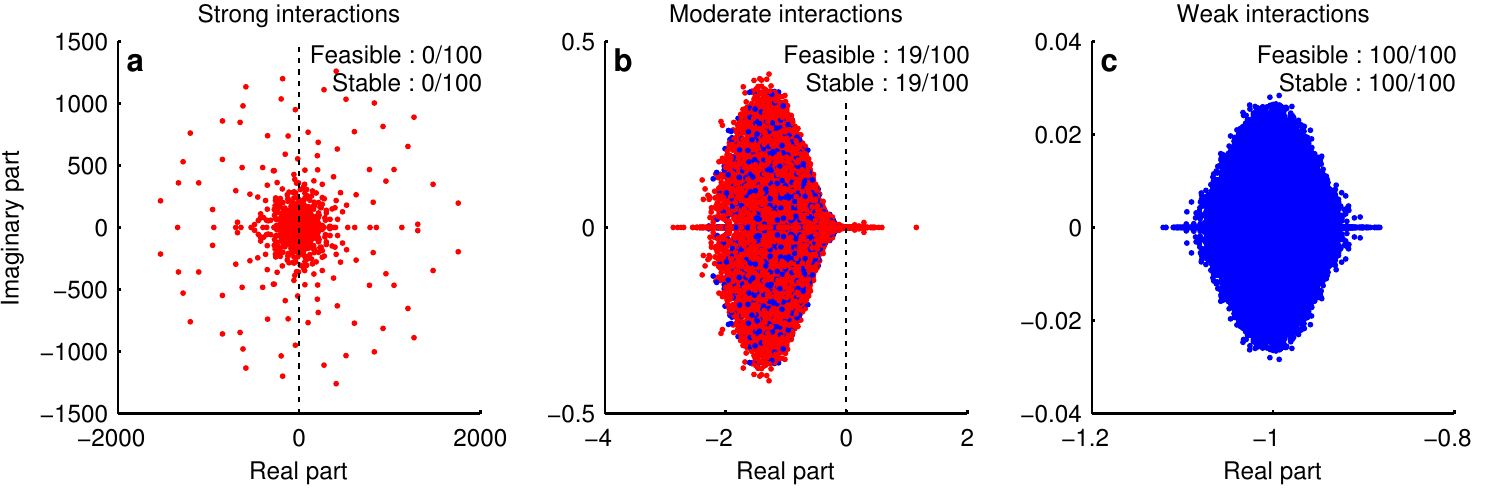}
\caption{Illustration that only feasible equilibria can be locally stable. For the three regimes of $\delta$, the eigenvalues of $J(x^*)$ are plotted in blue when $x^*$ is feasible, and red otherwise. (\textbf{a}) With strong interactions, equilibria are never feasible nor stable. (\textbf{b}) With moderate interactions,  feasible equilibria are stable. (\textbf{c}) With weak interactions, all equilibria are feasible and May's criterion is trivially fulfilled. The graphs show the eigenvalues of 100 realisations of the matrix $J$ for the random model with $S=150$, $\theta=-1$ and $\sigma=0.4$.}\label{fig:main1}
\end{figure}

\subsection{Feasibility in unstructured random web models}

%For strong interactions, as the size of the matrix $A$ grows, the variance of the equilibrium is so extreme that is is not logical to study  feasibility. Moreover, when considering large, complex networks, these interactions violate May's stability criterion\cite{May1972}.\\
For strong interactions ($0\leq\delta<1/2$), the probability of finding feasible equilibria $P_S$ goes abruptly to zero. Indeed, the variance of the abundances at equilibrium of each species grows with species richness. Then, the probability of having negative abundances converges to one, consequently $P_S \to 0$. Note also that this regime of interactions violates May's stability criterion when $S$ increases.

Moderate interactions ($\delta=1/2$) corresponds to the limiting case for which May's criterion can be asymptotically satisfied. Indeed, from equation (4), May's stability criterion is now independent from $S$ and $C$, and becomes $\sigma<|\theta|$. However, we show that in this case there exists asymptotically almost surely no feasible equilibria.
Indeed, we prove that the equilibrium abundances $x_i^*$  are asymptotically independent, identically distributed (i.i.d.)~Gaussian random variables
$$x^{*}_i \approx \mathcal{N}\lp \mu^*, \sigma^*\rp,\ i=1,\dots,S,$$
for some mean $\mu^*$ and some standard-deviation $\sigma^*$ (see Proposition~\ref{prop:asympt:solut:sublinear} in the Appendix). The probability of feasibility $P_S = P(x_i^*>0, \ \hbox{for all species}\ i)$ reduces thus by independence of the $x_i^*$ to $P_S = \prod_{i=1}^S P(x_i^*>0)$ in the large $S$ limit. Since every $x_i^*$ are identically normally distributed, $P(x_i^*>0)<1$ for every $1\leq i\leq S$. Thus $P_S$  decreases following a power law when the number of species $S$ becomes large, i.e.~$P_S \to 0$ as $S\to\infty$ (Fig.~\ref{fig:main2} (a). See Theorem~\ref{thm:feasibility:delta1/2} in the Appendix for complete details).

\begin{figure}[H!]
\includegraphics{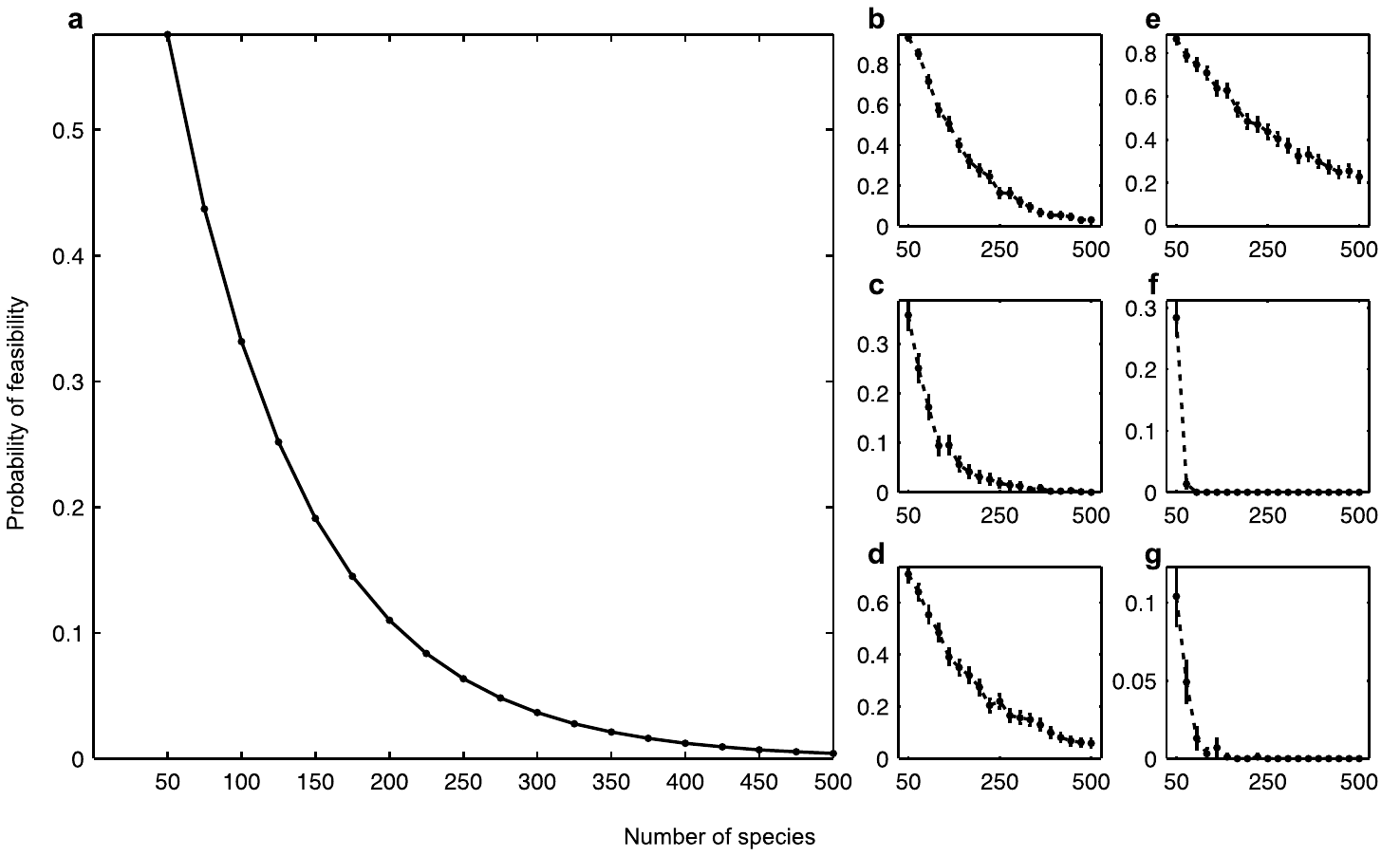}
\caption{Probability of feasibility for different models in the regime of moderate interaction strength. (\textbf{a}) Analytical prediction for the random model; (\textbf{b-g}) predictions (with 95\% confidence intervals) from 1000 simulations for random mutualistic networks, random competitive networks, random predator-prey networks, the cascade model~\cite{Cohen1990community}, the niche model~\cite{Williams2000simple}, and the nested-hierarchy model~\cite{Cattin2004}, respectively. We choose $a_{ij}\sim\mathcal{N}(0,\sigma)$ for any non-zero entry of the interaction matrix $A$ in the random model; in the other cases, a strictly positive interaction is randomly drawn from a folded normal distribution such that $a_{ij}\sim\left|\mathcal{N}(0,\sigma)\right|$, and a strictly negative interaction is sampled such that $a_{ij}\sim-\left|\mathcal{N}(0,\sigma)\right|$. The parameters are $C=0.25$, $\sigma=0.4$, and $\theta=-1$.}\label{fig:main2}
\end{figure}

%One then uses this result to obtain that
%\textcolor{orange}{
%$$0\le \lim_{S\to\infty}P_S\le 
%\Phi\lp\frac{\bar\mu}{\bar\sigma} \rp^{k},$$}
%\textcolor{blue}{
%for all $k\geq1$, where $\Phi$ denotes the standard Gaussian cumulative distribution function
%$$\Phi(z)=\frac{1}{\sqrt{2\pi}}\int_{-\infty}^z \exp(-\frac{u^2}{2}){\rm d}u.$$
%Since the probability on the right-hand side is strictly less than 1,  $P_S \to 0$ as $S\to\infty$.
%Hence, all the entries of $x^*$ are independent of each other, so that $P_S$  decreases following a power law when the number of species becomes large (Fig.~2a). }

We show numerically that
$P_S \to 0$ as $S\to\infty$   for unstructured random  models for  competitive, mutualistic, and prey-predator networks (see Table~1, Fig.~\ref{fig:main2}, and the Appendix).\\

\begin{table}[h!]
\begin{center}
\renewcommand{\arraystretch}{1.5}
\begin{tabular}{c l l l l }
\hline
& \multirow{2}{*}{\textbf{Model}} 		& \textbf{Moderate interactions}		& $\quad$ 	& \textbf{Weak interactions} \\
&					& \multicolumn{1}{c}{($\delta=\frac{1}{2}$)}			&		&  \multicolumn{1}{c}{($\frac{1}{2}<\delta\leq1$)}\\
\hline
\multirow{4}{*}{\begin{sideways} Unstructured~ \end{sideways}} 
& Random 		& $P_S\to 0$ (Thm.~\ref{thm:feasibility:delta1/2})	& $\quad$  	& $P_S\to0$ or $1$ (Thm.~\ref{thm:feasibility:delta1})\\

& Competition 		& $P_S\to 0$								& $\quad$ 	& $P_S\to0$ or $1$ (Cor.~\ref{cor:feasibility:delta1:mutPred})\\

& Mutualism 		& $P_S\to 0$								& $\quad$ 	& $P_S\to0$ or $1$ (Cor.~\ref{cor:feasibility:delta1:mutPred})\\

& Predation 			& $P_S\to 0$								& $\quad$ 	& $P_S\to0$ or $1$ (Cor.~\ref{cor:feasibility:delta1:pred})\\
\hline
\multirow{4}{*}{\begin{sideways} Structured~ \end{sideways}} & Cascade 			& $P_S\to 0$								& $\quad$ 	& $P_S\to0$ or $1$\\
& \multirow{2}{*}{Niche} 				& \multirow{2}{*}{$P_S\to 0$}			& \multirow{2}{*}{$\quad$	} 	&   $x^*$ not deterministic    \\[-0.75em]
&								&								&						&   and $P_S\to0$ or $1$    \\

& \multirow{2}{*}{Nested hierarchy}		& \multirow{2}{*}{$P_S\to 0$}			& \multirow{2}{*}{$\quad$} 	&   $x^*$ not deterministic   \\[-0.75em]
&								&								&						&   and $P_S\to0$ or $1$ \\
\hline
\end{tabular}
\caption{Summary of the different results presented on $P_S$ for $S\to\infty$. $P_S$ converges towards 0 as $S\to\infty$ for moderate interactions (and with the mean structural vector). For weak interactions, $P_S\to 0$ or $P_S\to 1$, depending on the parameters. The equilibrium $x^*$ is deterministic in the unstructured case and for the cascade model. $x^*$ has a non-trivial distribution and is feasible with positive probability for the niche and nested-hierarchy models. The results are analytical if specified in brackets and otherwise obtained by simulations (see sections~\ref{sec:simu} and~\ref{sec:simu:weak}).}\label{table:summary:results}
\end{center}
\end{table}

%\subparagraph{Weak interactions}
Consider now weak interactions among species ($1/2<\delta\leq1$). When the growth rate vector is set to the mean-structural vector, then $P_S\to1$ (Fig.~\ref{fig:main3} (a)). Indeed, we prove in the Appendix that the steady state $x^*$ converges towards the feasible deterministic vector $\mathbbm{1}$ which contains all 1s as values. In other words, any $x_i^*$ converges to 1 with a variance converging to 0 and is thus trivially positive. This implies $P_S \to 1$.

% Fig. weak interactions

\begin{figure}[H!]
\includegraphics{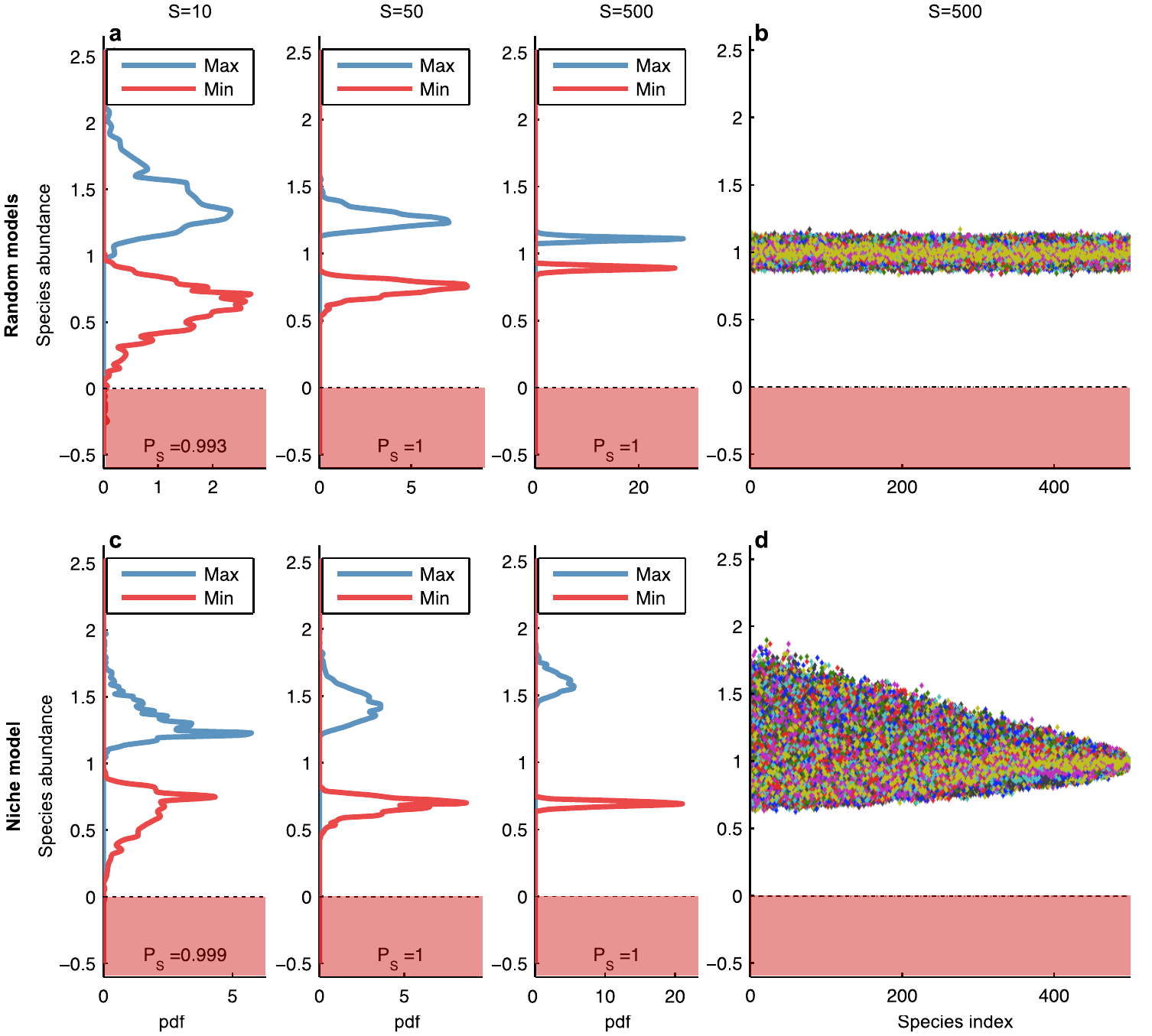}

\caption{Convergence of the equilibrium under weak interactions. For $S=10$, 50, and 500 the distribution of $\textrm{max}(x^*)$, in blue, and $\textrm{min}(x^*)$, in red, are represented for 1000 simulations for (\textbf{a}) random networks, and (\textbf{c}) the niche model. For $S=500$, all simulations are plotted to illustrate the support of the distribution of $x^*$: (\textbf{b}) in random networks, the regime of weak interactions yields asymptotically deterministic equilibria; (\textbf{d}) in the niche model, the equilibria are no longer deterministic and converge towards a distribution with a non-trivial support. Species range from top predator (index $i=1$) to basal ($i=S$); in (\textbf{d}), note that the size of the support depends on the index of the species.}\label{fig:main3}
\end{figure}

We test now different growth rates. Considering a general model where $\mu_{A}$ denotes the mean interaction strength between interacting species before normalisation, we find that $P_S \to 1$ when $r_i>\bar r\,\mu_{A}/(\mu_{A}+\theta)$ for all $i$ and where $\bar r$ is the sample mean of the intrinsic growth rates (see Theorem~\ref{thm:feasibility:delta1} in the Appendix), otherwise, $P_S \to 0$. For example, in May's framework ($\mu_{A}=0$), $P_S \to 1$ when all growth rates are positive. Similar results will be found in the Appendix concerning other types of unstructured models. These results rely on the fact that, when the growth rate vector has bounded components, the variance of any $x_i^*$ goes towards zero. As such, $x^*$ converges towards a deterministic vector that is feasible or not depending on the chosen parameters values. Thus $P_S\to1$ or $0$, depending on $r_i$, $\theta$, and $\mu_{A}$.
%In the May's framework, $P_S \to 1$ when all growth rates are positive. As soon as one growth rate is negative, then $P_S \to 0$. Similar results will be found in the Supplementary Informations concerning other types of unstructured models. These results all rely on the fact that when the growth rate vector has bounded components, the variance of any $x_i^*$ goes towards zero. As such, $x^*$ converges towards a deterministic vector which is feasible or not depending on the chosen parameters values. Thus $P_S\to1$ or $0$, depending on the $r_i$, the intraspecific competition and the mean interaction strength $\mu_{A}$ between species (see Theorem S.3.6.~of the Supplementary Informations. The basic relation is that $P_S\to1$ when for all $i$, $r_i>\mu_{A}/(\mu_{A}+\theta)\,\bar r$, where $\bar r$ is the mean intrinsic growth rate).}

At all, with weak interactions and for any choice of growth rates, all realisations of the random matrix $A$ lead to a limiting deterministic equilibrium. Thus, abundance variability disappears irrespective of particular realisations of $A$.

% Moreover, when the growth rate vector is set to the so-called mean structural vector, $x^*$ converges towards the feasible vector $\mathbbm{1}$ which contains all 1s as values, and thus $P_S\to1$ (Fig.~3a). For this model, one can show that any choice of positive growth rate vector $r$ leads to an almost surely feasible equilibrium.

When studying the three other types of unstructured models (mutualistic, competitive, and predator-prey networks), we find the same results as above for strong, moderate or weak interactions (Table~1). 

%%%%%%%%%%%%%%%%%%%%%%%%%%%%%%%%%%%%%%%%%%%%%%%%%%
\subsection{Feasibility in predator-prey structured models.}

 May's paradox that complexity decreases stability in mathematical models led to the exploration of the effects of topological structure on systems' dynamics. Several works have shown that particular architectures can stabilise ecological networks~\cite{DeAngelis1975,Dunne2002,Okuyama2008,Thebault2010}. Nevertheless, the feasibility of different network architectures for different interaction regimes has not been explored.

% Cascade model: moderate + weak interactions
Consider for example Cohen's cascade model of trophic networks~\cite{Cohen1990community}, where species are hierarchically ordered so that they feed only on lower indexed species, generating sub-triangular matrices with random structure. For the three regimes of $\delta$, simulations yield similar behaviour of $P_{S}$ as for unstructured models (Table 1).

We now test two other food-web models that more accurately capture the structure of trophic interactions in real systems, the niche~\cite{Williams2000simple}  and the nested-hierarchy models~\cite{Cattin2004}. The former generates purely interval food webs (predators consume all species in a niche interval) while the latter is based on evolutionary processes and relaxes this constraint. The results are qualitatively similar to the ones of unstructured random models for all types of interactions. However, two interesting features emerge. First, the mean structural vector now contains negative growth rates, which is the essence of predator-prey systems where predators will die in the absence of prey (this feature also appears in the cascade model, but not in the unstructured food-web model, see Table~2 in Appendix). Second, with weak interactions, each $x_i^*$ no longer converges to a deterministic limiting vector, but instead to a non-trivial random variable having a distribution of finite support and positive standard-deviation. Thus, contrary to the unstructured and cascade models, a particular realisation of the interaction matrix $A$ will influence the equilibrium values. Indeed, these highly structured networks and their related mean structural growth rates induce correlations among the abundances at equilibrium, which prevent a convergence to a deterministic value. We find that the size of the support depends on the hierarchical position of the species (Fig.~\ref{fig:main3} (c-d) and Fig.~\ref{fig:delta1:structured}). We performed additional simulations based on empirical food webs, and find similar results (Figs.~\ref{fig:delta1:structured:histogramNestedNiche}, \ref{fig:nicheNested:rRandom:weak:variances} and \ref{fig:empirical:meanStructuralVector}).

Since the mean structural vector is constructed in the purpose of maximizing the probability of feasibility, we have also considered random Gaussian growth rates of positive mean to complete our analysis.  For the niche and nested-hierarchy models, the standard deviations $\sigma^*_i$ of the abundances at  equilibrium are ordered according to the trophic position of the species within the food web (from basal autotrophs, to intermediate species that are predator and prey, and to top predators). We find $\sigma^*_{\rm basal}<\sigma^*_{\rm intermediate}<\sigma^*_{\rm top}$, a result similar to that obtained with the structural vector (Figs.~\ref{fig:delta1:structured:histogramNestedNiche} and \ref{fig:nicheNested:rRandom:weak}). This effect, also visible in empirical food-webs (Figs.~\ref{fig:nicheNested:rRandom:weak:variances} and \ref{fig:empirical:meanStructuralVector}), disappears in two particular cases: 1) for zero mean Gaussian random growth rates; 2) when the growth rate standard-deviation becomes large. In both cases, the equilibria then become independent from species trophic position (Figs.~\ref{fig:nicheNested:rRandom:weak:variances} and \ref{fig:empirical:randomGrowthVector}).

\section{Discussion}

Our results can be summarised as follows: for the cases compatible with May's~\cite{May1972} and Allesina's~\cite{Allesina2012,Allesina2015} framework where $\delta = 1/2$ (moderate interactions), there exist almost surely no feasible equilibria in species-rich systems for any of the models considered. For weak interactions in randomly structured systems, feasibility is granted; also, the stability criteria are trivially asymptotically satisfied. Therefore, in such situations, the criteria are basically void of information. However, including realistic structure in trophic systems, we find that the results are no longer trivial. 

As a general message, future works tackling dynamical stability must be preceded by a feasibility analysis. In our Lotka-Volterra framework, this involves knowledge of the structure and intensity of interactions, and of the growth rates. Feasibility can then be easily evaluated. Our results indicate that sensible information on the dynamics of natural systems will be obtained by concentrating on the interplay between ecological structure, strengths of interaction weights, and how interaction strengths are distributed in their architecture. 

There are still many other aspects that should be taken into account to reach general conclusions about feasibility and stability that better fit natural systems. They are very different from unstructured systems, and many constraints affect their architecture and dynamics. A first question simply concerns system size in term of species richness $S$. Most described systems are middle sized, but large webs do exist~\cite{OGorman2010,Twomey2012}, so that mathematical study focusing on large $S$ limits can be justified. An intriguing aspect here is how smaller and well-definable webs are dynamically embedded in larger species-rich systems. For example, it has been found that top predators have a stabilising effect on food webs by coupling fast and slow energy channels in natural aquatic systems~\cite{Rooney2006}. This question could be framed in a general theory for the ``inverse pyramid of habitat", which explores the consequences of the ubiquitous (but often overlooked) observation that species located at higher trophic levels tend to have larger home ranges~\cite{Holt2005}. Relevant to system size lies also the question of how network nodes are defined: taxonomic resolution is often heterogeneous in observed food webs, with basal species being more often pooled (e.g., phytoplankton being considered as a single species). Secondly,  species richness $S$ and connectance $C$ are usually treated as independent parameters, while it has been described that $C$ decreases with $S$ in different systems~\cite{BanasekRichter2009}, thus restricting the parameter space to be explored. Finally, even if topology per se has been suggested to play a minor role in system stability~\cite{Johnson2014,James2015}, we showed that network architecture does play a non-trivial role on equilibrium abundances. Layer architectures as those produced by the niche and nested-hierarchy models yield ecologically sensible results; exploring other structural models~\cite{Drossel2004,Allesina2008,Rossberg2008} and especially natural architectures is necessary. 

Apart from these structural considerations, ecologists are confronted to a more elusive issue, the estimation of demographic parameters. To make sense, the investigation of feasibility and stability should be based on meaningful parameter values. The growth rate vector $r$ that appears in Lotka-Volterra dynamics is mostly unknown or very difficult to obtain experimentally. One can either keep this vector as an intrinsic free parameter or assume that such growth rates change with interaction weights, as in~\cite{James2015} or with our mean structural vector. The estimation of interaction parameters is also far from trivial~\cite{Laska1998,Berlow2004,Arditi2016}. Experimental studies (see, e.g.,~\cite{Neutel2014,James2015,Jacquet2013}) choose interaction weights indirectly, according to particular methods like predator-prey mass ratio models, biomass flux, or the Ecopath method (see, e.g.,~\cite{Emmerson2004B,Moore1996,Christensen1992}). Additionally, one unresolved question is how intraspecific competition scales with average interaction strength, which is perhaps the most important ingredient for stability (e.g.,~\cite{May1972,Neutel2007,Allesina2012,Neutel2014}), and also plays a role in feasibility (see Proposition~\ref{prop:asympt:solut:sublinear} in Appendix). Also, if the magnitude of intraspecific competition can be assumed to be independent from $S$, this is likely not the case with interspecific interactions (captured by our parameter $\delta$), which has be shown to play a key role in mutualistic networks~\cite{Rohr2013}. This question is related to the study of weak interactions, which have been shown to promote stability~\cite{McCann1998, Neutel2014, Jacquet2013,Gellner2016} and here feasibility. Finally, one must consider that these parameters do scale with body size~\cite{Brose2006}, which defines a constraint between parameter values and network structure, as large-bodied species typically populate higher trophic levels. Interestingly, this allometric relationship may underly the result that stability depends on the link between trophic position and interaction strength~\cite{James2015}, and our finding that abundance variability scales with trophic position. 

For mathematical tractability, many studies on the dynamics of food webs, including ours, rely on Lotka-Volterra models with mass-action type interactions (the so-called Holling type I functional response). However, investigations on the number of prey eaten per predator and unit time in simple systems do not support this modelling assumption. There exist a vast literature on the subject and the choice of a sensible functional response is still debated~\cite{Arditi2012Book,Abrams2015}. Also, it has been shown that ``foraging adaptation" by predators affects the dynamics of food webs~\cite{Kondoh2003}. Such non-linearities in species interactions have been explored from the stability point of view in systems adjusted to be feasible~\cite{Gross2004,Gross2006,Gross2009}. However, since the precise way interactions are modelled has non-trivial consequences on system dynamics, the relationship between feasibility and complexity should also be considered in such situations. 

May's approach considered systems with any kind of interaction and not only trophic ones. An articulate answer on the relationship between stability and complexity should obviously incorporate all types of interactions relevant for system dynamics (including the effects of ecosystem engineers~\cite{Jones1994,Odling2003}). This question has been the focus of recent theoretical developments~\cite{Arditi2005,Fontaine2011}, and empirical studies specifically addressing this question start to emerge~\cite{Pocock2012}.   

Science faces increasingly complex situations where high-dimensional parameters occur. A good example is statistical mechanics, for which relevant results have been obtained for highly complex systems without having a precise knowledge of the microscopic details of a model. Our results like many cited here follow the same line. However, due to the complexity of natural systems, the gap between theoretical and empirical investigations is likely to remain open. To attain a consensus, the feasibility of the model systems should not be forgotten~\cite{vandermeer1970,Logofet1993Book,Bastolla2005,Nattrass2012,Rohr2013}, and more effort must be devoted to obtain empirical and experimental time series as necessary benchmarks.

\appendix

%\section{Summary of the results}

\section{Model and methods}\label{sec:model}

%------------------------------------------------------------------------------------------
\subsection{Feasibility and stability in dynamical systems}

Consider dynamics driven by a general system of first order autonomous ordinary differential equations
$$\frac{\textrm{d}x_i(t)}{\textrm{d}t} = f_i\lp x(t) \rp, \quad i \in {\cal S} = \left\{1,\dots,S \right\}, \quad t\geq0,$$
with $f=(f_1,\ldots,f_S):\R^S\to\R^S$ a differentiable function and $x(t) = \lp x_1(t),\dots,x_S(t) \rp\in\R^S$. Solutions $x(t)$ of this system evolve in time and their trajectory are rarely describable as such, especially in high dimensions. However, there are some important features that can be exhibited without solving the system. One of the most instructive is the existence of \emph{equilibria}, which are points $x^*$ in the phase space $\R^S$ where the solution of the system does not vary in time, that is $f_i(x^*) = 0$ for all $i \in {\cal S}$. A trajectory reaching an equilibrium remains indefinitely at this point. If $x^{*}$ belongs to the admissibility domain of the model (i.e., all $x^*_i>0$ for ecological systems), then it is called \emph{feasible}. Unfeasible equilibria are well-defined in the mathematical sense, but have to be rejected in the perspective of the application.

If small perturbations of the trajectory in the neighbourhood of the equilibrium $x^*$ fade over time, so that the system tends to restore the equilibrium, then  $x^*$ is said to be \emph{locally stable}. This property is related to the derivatives of the function $f$. More specifically, an equilibrium $x^*$ is \emph{linearly stable} when the Jacobian matrix $J(x^*)=\lp\frac{\partial f_i }{\partial x_j}\,(x^{*})\rp$ evaluated at the equilibrium has only eigenvalues with negative real parts (see~\cite{May2001Book},~\cite{Tabor1989Book}  or~\cite{Logofet2005}).

%------------------------------------------------------------------------------------------
\subsection{Lotka-Volterra model}
In the context of ecological networks, the dynamics of interacting species is commonly described by the Lotka-Volterra equations (see e.g.~\cite{lotka1925,volterra1926,Gillman1997}). The function $f$ can be written as
\begin{equation}\label{LVWM}
\frac{{\rm d}x_i}{{\rm d}t} =  x_i \lp r_i+\theta x_i+\sum_{j=1}^S \frac{a_{ij}}{(CS)^{\delta}}\,x_j\rp = f_i(x),\ i\in {\cal S}=\{1,\cdots,S\},
\end{equation}
where $x_i$ denotes the abundance of species $i$,  $\theta < 0$ is a friction coefficient (intraspecific competition, assumed to be the same for all species), $r_{i}$ describes the intrinsic growth rate of $i$ and the interaction coefficients $a_{ij}$ stand for the per capita effect of species $j$ on species $i$. The connectance $C$ denotes the proportion of links present in the network with respect to the number of all possible links $S(S-1)$ or $S^2$, depending on the model. The product $CS$ is a measure of the complexity of the system, and is the average number of links between two species. The term $(CS)^{\delta}$ is introduced as a normalisation of the interactions strength, with $\delta\geq0$ a parameter controlling this renormalisation (see 2.3.3 below). In matrix form, the system writes
\begin{equation}\label{eq:LVWM:matrix}
\frac{{\rm d}x}{{\rm d}t} =  x\circ\lp r+\lp \theta\, {I} + \frac{A}{(CS)^{\delta}}\rp x\rp,
\end{equation}
where $I$ denotes the $S\times S$ identity matrix, $A=(a_{ij})_{1\leq i,j\leq S}$ is the interaction matrix,   $r=(r_1,\ldots,r_S)$ is the growth rates vector and $\circ$ denotes the Hadamard product, that is $x\circ y=(x_1y_1,\ldots,x_Sy_S)$.

The admissibility domain of this model is the positive orthant, that is $x\in\R^{S}$ with $x_{i}>0$ for all $i$. %For simplicity reasons, we do not consider points with extinguished species (such that $x_i=0$) as admissible. This latter assumption can however easily be removed by restricting the network to the species with positive abundance. 
A feasible equilibrium must then satisfy $x^{*}_i>0$ for all species $i\in\mathcal{S}$ and 
\begin{equation}\label{eq:def:x*}
x^{*} =\lp-\theta\, {I}-\frac{A}{(CS)^\delta}\rp ^{-1} r,
\end{equation} 
so that the right hand side of~\eqref{eq:LVWM:matrix} equals zero, assuming this inverse matrix exists. Such an equilibrium is locally stable if the Jacobian matrix (also named community matrix in this context)
\begin{equation}\label{eq:def:J}
J\lp x^{*}\rp={\rm diag}\lp x^{*}\rp\lp\theta\, {I}+\frac{A}{(CS)^{\delta}}\rp
\end{equation}
has all eigenvalues with negative real part. As one can see in the previous equation, this matrix depends explicitly on the equilibrium abundances $x^*$  (${\rm diag}\lp x^{*}\rp$ denotes the matrix with $x^{*}$ in the diagonal and $0$ everywhere else) and it should not be confused with the interaction matrix $A$.\\

Several ecological models have been developed with the goal of understanding network structure and its effects on system dynamics~\cite{Rossberg2013book,May1972,Cohen1990community,Williams2000simple,Cattin2004}. Those models generate essentially different types of network structures, and consequently particular interaction matrices $A$. In order to exhibit the particular effects of a structure on the dynamics of systems, the coefficients $a_{ij}$ can be randomised. This allows the exploration of some behaviour of the models and the detection of their key features. A direct consequence in considering random interactions is that feasibility and stability have to be considered from a probabilistic point of view. We define therefore the \emph{probability of feasibility} of an equilibrium as
\begin{equation}\label{eq:def:p}
P_{S} = \Prob\lp x_i^* > 0 \, \forall i \in {\cal S} \rp.
\end{equation}
The purpose of this work is to study this probability, which expresses the likelihood of a model to provide feasible equilibria, for which the stability can eventually be examined.

%------------------------------------------------------------------------------------------
\subsection{Interactions}
The interaction matrix $A$ describes firstly who interacts with whom in the network (the structure), and secondly what is the type of these interactions: if $a_{ij}<0$ and $a_{ji}<0$, species $i$ and $j$ compete with each other; if  $a_{ij}>0$ and $a_{ji}>0$, their interaction is mutualistic and they both benefit from the presence of each other; if $a_{ij}>0$ and $a_{ji}<0$, then species $i$ preys upon species $j$; if $a_{ij}=0$, species $j$ has no direct effect on $i$.

The models of complex ecological networks that we consider can be divided into two main categories: unstructured and structured models. The former consist in webs for which the topological structure is completely free and random. This equates to consider networks based on Erd\H{o}s-R\'enyi graphs, which are constructed by adding randomly edges between nodes with the same probability $C$. Here, we consider four of those models depending on the type of interactions: random (as in May's formalism~\cite{May1972}), mutualistic, competitive and predator-prey. Structured models define stochastic rules for the construction of the underlying graphs so that not all graphs are equally likely. This is the case in the three models studied here: the cascade~\cite{Cohen1990community}, the niche~\cite{Williams2000simple} and the nested-hierarchy model~\cite{Cattin2004}. Since these three models were built to represent food webs, only predator-prey interactions are considered. Fig.~\ref{fig:networks} provides an example of network samples for the unstructured predator-prey model and for the structured cascade and niche model.

%----------------------------------------------------
\subsubsection{Interaction models on unstructured networks}\label{sec:interactionsModels}

\textbf{Random network}. This model introduced by May in 1972~\cite{May1972} considers both the structure and the type of interactions as fully random. The interactions are independent and identically distributed (i.i.d.)~centred random variables with common standard deviation $\sigma$ and can thus be either positive or negative,  corresponding to a mixture of competition, mutualism and predation. The coefficients equal $0$ independently of their position, with a fixed probability $1-C$. Note that $a_{ij}=0$ does not necessary imply $a_{ji}=0$, so that amensalism $(0,-)$ and commensalism $(0,+)$ interactions are also possible. In this model, the connectance is given by $C=L/S^2$, where $L$ is the number of links in the network.

\medskip
\noindent
\textbf{Mutualistic network}. This model considers only mutualistic interactions $(+,+)$. The diagonal coefficients are $a_{ii}=0$ and any pair of species $(i,j)$, $i< j$ is linked with probability $C$. If the pair $(i,j)$ is not linked, then $a_{ij}=a_{ji}=0$, otherwise the interactions strength $a_{ij}$ and $a_{ji}$ are positive i.i.d.~random variables. The elements of $A$ are almost all i.i.d., except for the pairs $a_{ij}$ and $a_{ji}$, which are only independent conditionally on the fact that they are different from $0$. In this model, the connectance is given by $C=L/S(S-1)$.

\medskip
\noindent
\textbf{Competitive network}.  This model is built in the exact same manner as the previous one, except that the interactions are nonpositive random variables in such a way that there are only competitive interactions $(-,-)$.

\medskip
\noindent
\textbf{Predator-prey network}. This is an example of an unstructured model for \emph{predation}. The diagonal coefficients are $a_{ii}=0$ and any pair of species $(i,j)$, $i< j$ is linked with probability $C$. If the pair $(i,j)$ is not linked, then $a_{ij}=a_{ji}=0$, otherwise the interactions strength are independently sampled, with the restriction that $\sign(a_{ij}) = -\sign(a_{ji})$. This results in a sign antisymmetric interaction matrix whose elements are identically distributed and almost all independent, except for the pairs $(a_{ij}, a_{ji})$, which are correlated by their sign. The connectance in this model is given by $C=L/S(S-1)$.

\subsubsection{Interaction models on structured networks}

\textbf{Cascade model}. This model introduced by Cohen et al.~\cite{Cohen1990community} is an example of structured food webs. The interactions are of predator-prey type. The species are ordered on a line and they can feed only on species with a strictly lower rank, excluding any loop in the network. The resulting interaction matrix $A$ has an upper diagonal with nonnegative i.i.d.~entries and a lower diagonal with nonpositive i.i.d.~entries. An upper diagonal entry $a_{ij}, i<j$ can be zero with probability $1-C$ and in this case $a_{ji}=a_{ij}=0$. The diagonal of $A$ consists in zeros and the connectance is $C=L/S(S-1)$.

\medskip
\noindent
\textbf{Niche model}. In 2000, Williams and Martinez~\cite{Williams2000simple} proposed a new model that permits loops (including cannibalistic loops), and that generates purely interval food webs. Each species is randomly assigned three numbers: a niche value, a range radius proportional to the niche value and a range centre. Species feed on all species whose niche value falls into their range. The species with the smallest niche value has a range $0$ to ensure the presence of, at least, one basal species. This defines the so-called adjacency matrix of the network: $\alpha$ with $\alpha_{ij}=1$ if $i$ preys upon $j$ and $\alpha_{ij}=0$ otherwise. The parameters of the probability distributions of the niche and the range are chosen in order to obtain an average connectance $C/2$ in the matrix $\alpha$. For the construction of the corresponding interaction matrix, one multiplies the entries of $\alpha$  with positive i.i.d.~values and the entries of $\alpha^{t}$ with negative i.i.d.~values, and finally one adds up the two corresponding matrices. This results in a interaction matrix $A$ with upper diagonal with mostly (but not exclusively) nonnegative entries, and lower diagonal entries mostly nonpositive. The average connectance of $A$ is $C=L/S^{2}$.  Notice that this construction is slightly different from the one proposed by Allesina and Tang in~\cite{Allesina2012}, since if species $i$ preys on $j$ and vice-versa, those interactions do not cancel out but are added in $A$. The combination of both interactions can either be mutualistic or competitive, or can remain a prey-predator interaction.

\medskip
\noindent
\textbf{Nested-hierarchy model}. Cattin et al.~\cite{Cattin2004} proposed a model that tries to implicitly take evolution into account and relaxes the intervality of the diets of the niche model. The model creates a nested hierarchy between species driven by phylogenetic constraints. Two numbers are randomly assigned to each species: a niche value and a number of preys that depends on the niche value. The species with the smallest niche value has no prey. After having reordered the species according to their niche value, preys are assigned to species starting with the species that has the smallest niche value, according to the following rule: for a consumer $i$, one chooses randomly a prey $j$ with smaller niche value; then, one considers a pool of preys consisting of all preys consumed by other consumers of $j$; consumer $i$ will then be assigned preys among this pool; if the pool is too small, choose another pool, and if this is still not possible, choose randomly a new prey. This algorithm builds the adjacency matrix of the network $\alpha$, whose connectance is set on average to $C/2$ by tuning the parameters of the probability distribution of the niche values and the number of links of each species. The interaction matrix $A$ is then constructed in the exact same way as for the niche model and the connectance is $C=L/S^{2}$.

\subsubsection{Interaction strength}
Intensity of interactions plays a key role for the local stability of networks~\cite{McCann1998,Emmerson2004}. As already suggested by May's criterion~\cite{May1972}, in the framework of the random model, interactions strength should decrease at least at rate $\frac{1}{\sqrt{CS}}$ in order to preserve stability when complexity $CS$ increases. This consideration motivates the introduction of the normalising parameter $\delta\geq0$ in~\eqref{eq:LVWM:matrix} for the study of ecological feasibility of equilibria. Three regimes naturally emerge: strong interactions ($\delta=0$), moderate interactions ($\delta=1/2$) and weak interactions ($\delta=1$). A useful approach for the interpretation of these different regimes is the weight of a node in the network. Consider for example the random model. As the network is unstructured, all species have on average the same role, and the expected total weight of their interactions is
$$W_{i}=\frac{1}{(CS)^{\delta}}\,\sum_{j=1}^{S}\E\lp|a_{ij}|\rp=\frac{1}{(CS)^{\delta}} \,S\,\E\lp|a_{11}|\rp=(CS)^{1-\delta}\, \E\lp|a_{11}|\mid a_{11}\neq0\rp$$
since $a_{ij}$ are i.i.d.~by assumption. One sees that for $\delta=1$, $W_{i}$ is constant and does not depend on the complexity. For $\delta=\frac{1}{2}$, the total weight is proportional to the square root of the complexity, and for $\delta=0$, $W_{i}$ depends linearly on the complexity. From the mathematical point of view, these three regimes correspond to three normalisation of the sum of i.i.d.~random variables. Indeed, let $X_{1},X_{2},\ldots$ be a sequence of i.i.d.~centred random variables. The rescaled sum $\frac{1}{n^{\delta}}\sum_{i=1}^{n} X_{i}$ converges almost surely to zero if $\frac{1}{2}<\delta\leq1$ (law of large numbers). If $\delta=\frac{1}{2}$, the central limit theorem states that $\frac{1}{n^{\delta}}\sum_{i=1}^{n} X_{i}$ converges in distribution to a normal distributed random variable with mean zero and finite standard deviation. Proposition~\ref{prop:asympt:solut:sublinear} and Theorem~\ref{thm:feasibility:delta1} below show the link between these probabilistic regimes and the nature of the equilibrium as the size of the network goes to infinity, depending on the parameter $\delta$. While $x^{*}$ converges to a deterministic value when interactions are weak ($\delta=1$), it obeys a central limit theorem and converges to a well-defined random variable when the interactions are moderate ($\delta=\frac{1}{2}$). Finally, if $0\leq\delta<\frac{1}{2}$, the rescaled sum does not converge and its standard deviation explodes.  Consequently, the case of strong interactions ($\delta=0$) is not studied, as the mathematical problem is asymptotically ill-posed.

%------------------------------------------------------------------------------------------
\subsection{Growth rates}

The parameter $r$ is of particular interest when studying feasibility as illustrated in~\cite{Rohr2013}. For any fixed realisation of the random interaction matrix $A$, it is always possible to tune the parameter $r$ in order to get any desired equilibrium. Choose any vector $u\in\R^S$ and set 
$$r(u)= \lp- \theta I - \frac{A}{(CS)^{\delta}}  \rp u$$ 
as a structural growth rates vector. A glance at~\eqref{eq:def:x*} shows immediately that the resulting equilibrium is $x^*=u$. Consider now the all-ones vector $\mathbbm{1}$ in $\R^S$. This direction $\mathbbm{1}$ is in a sense the most feasible, since it is as far as possible from the boundaries of the admissibility domain (i.e.~the positive orthant). Setting the growth rates vector to the so-called structural vector $r(\mathbbm{1})$ leads the equilibrium deterministically to $x^*=\mathbbm{1}$. 
Interestingly, the authors of \cite{James2015} 
studied stability using non-random interaction weights. They used $r(u)$ with 
$u=(-d/\theta)\mathbbm{1}$, for some positive constant $d>0$,
to avoid transcritical bifurcations where the abundances of the steady state $x_i^*$ become negative. They argue that with this particular choice, the only way in which a species can become extinct is via a degenerate bifurcation that makes a positive equilibrium unstable, causing the system to move suddenly to a different equilibrium point.

However, choosing the parameter $r$ in such a way, i.e., contingently to $A$, is clearly uninformative as it erases the whole structure of the model, and nips the purpose of the randomisation in the bud. The parameter $r$ should thus not be set according to a particular realisation of $A$. This can be achieved with, for example, the \textit{mean structural vector} $v$ defined by
\begin{equation}\label{eq:def:rStructural}
v=\E \lp r(\mathbbm{1})\rp = \E \lp- \theta I  - \frac{A}{(CS)^{\delta}} \rp \cdot \mathbbm{1}.
\end{equation}
This vector is deterministic, therefore clearly independent of the randomness of $A$, and it provides equilibria that are, on average, the most feasible. Choosing this $v$ as growth rates vector allows then to study feasibility in the most favourable conditions. The mean structural vectors of the seven models of interest are given in Table~2.

\begin{table}
\begin{center}
\renewcommand{\arraystretch}{1.6}
\begin{tabular}{l c l l }
\hline
\textbf{Model} & \textbf{Mean interaction} & $\quad$ & \textbf{Mean structural vector} \\
\hline
Random 	&  									& $\quad$  	& $v_{i}= \left| \theta \right| - (CS)^{1-\delta}\, \mu_{A}$ \\

Competition 		& $\E(a_{ij}\mid a_{ij}\neq0) = \mu_{A}$ 			& $\quad$ 	& $v_{i}= \left| \theta \right| - \frac{C(S-1)}{(CS)^\delta}\, \mu_{A}$ \\

Mutualism 		&									& $\quad$ 	& $v_{i}= \left| \theta \right| - \frac{C(S-1)}{(CS)^\delta}\, \mu_{A}$ \\
\hline
Predation 			& \multirow{3}{*}{$\E(a_{ij}\mid a_{ij}>0) = \mu_{A_+} > 0$ and}	& $\quad$ 	& $v_{i}=\left| \theta \right| - (CS)^{1-\delta}\cdot \frac{\mu_{A_-} + \mu_{A_+}}{2}$\\

Cascade 			&  \multirow{3}{*}{$\E(a_{ji}\mid a_{ij}<0) = \mu_{A_-} < 0$\phantom{ and}}		& $\quad$ 	& $v_{i}=\left| \theta \right| - \frac{i-1}{(CS)^{\delta}}\, C \,\mu_{A_-} - \frac{S-i}{(CS)^{\delta}}\,C\, \mu_{A_+}$\\

Niche 			&									& $\quad$ 	& \textit{Simulated}\\
Nested hierarchy	&									& $\quad$ 	& \textit{Simulated}\\
\hline
\end{tabular}
\caption{Models of interactions and their corresponding mean structural vector $v$. Note that the form of $v$ is the same for the random, competition and mutualism models. In a competitive system, $\mu_{A}<0$ and in a mutualistic one, $\mu_{A}>0$. In all unstructured models (random, competition, mutualism, predation), the mean structural vector has the same direction as the all-ones vector, meaning that the structural growth rates are the same for all species. For the structured models (cascade, niche and nested hierarchy), species do have different roles in the network, while in the other models all species behave, on average, identically. The mean structural vector follows then a particular direction which is no more the all-ones vector. One sees that, for the cascade model, for appropriate choice of the parameters $\theta$, $\mu_{A_+}$ and $\mu_{A_-}$ leads to a vector $v$ which has negative entries for species with small indices $i$ (top predators) and positive entries for species with large indices (basal species). This type of growth rate vector is a characteristic of such networks~\cite{lotka1925,volterra1926}. For the niche and the nested hierarchy models, the algorithmic construction of the interaction matrix makes the formula of their respective $v$ very complicated. They are therefore obtained by Monte Carlo simulations. As for the cascade model, negative growth rates are observed.}\label{table:rs}
\end{center}
\end{table}

The analytical results of sections~\ref{sec:analytical:moderate} and~\ref{sec:analytical:weak} hold for more general growth rates vector than $v$. They allow to consider random vectors $r=(r_i)$ with i.i.d.~entries. In the case of unstructured models, $v$ is indeed a particular case of this type of vectors, as the entries are all identical and deterministic, therefore independent from each other. In the following, we always illustrate our results with the mean structural vector $v$, unless specified.

%%%%%%%%%%%%%%%%%%%%%%%%%%%%%%%%%%%%%%%%%%%%%

\section{Probability of feasibility}\label{sec:results}
We study analytically the probability of feasibility $P_S$ in the framework of the random model. For moderate interactions, $P_S$ decreases exponentially with the size of the system. For weak interactions, $P_S$ is asymptotically equal to one when the growth rate vector is set to $v$ (see Eq.~\eqref{eq:def:rStructural}). We show that for other choices of growth vector, if the interactions are weak, then $P_{S}$ either equals one or zero, depending on the value of the parameters. We provide estimations of $P_{S}$ by mean of Monte Carlo simulations for the other models and observe the same phenomenon. In the case of weak interactions, $P_{S}$ can be estimated analytically for all unstructured models. 

%For weak interactions, the  cases $\frac{1}{2}<\delta \leq1$ are all equivalent to $\delta=1$. When the size of the network increases, the equilibrium $x^*$ converges towards a deterministic point. In contrast to this, the equilibria obtained under a regime of moderate interactions ($\delta=\frac{1}{2}$) are subject to randomness. They are random variables whose variance do not go towards zero when $S$ is large. 

%------------------------------------------------------------------------------------------
\subsection{Moderate interactions}\label{sec:resultsModerate}
We show that in the random model with moderate interactions, there exists asymptotically almost surely no feasible equilibrium of the system~\eqref{LVWM}. The lack of structure of this model brings independence of the entries of $x^*$, the biomasses at equilibrium. Moreover, the rate of normalisation $\delta=\frac{1}{2}$ makes these entries follow asymptotically a normal law. Thus, when the number of species becomes large, there is a smaller and smaller probability that all of them have positive equilibrium, as illustrated in Fig.~\ref{fig:sim:proba:persist} and Fig.~\ref{fig:main1}.

This phenomenon is a key feature of this model and does not depend on the choice of the intrinsic rates $r$. Indeed, we can show analytically that the probability of persistence at equilibrium goes towards zero for any choice of a random vector $r$  with i.i.d.~entries. This framework contains the particular case where the intrinsic rates are all equal, i.e.~the same direction as $v$,  the mean structural vector of the random model. 

We extend this study by estimating $P_S$ in the other unstructured and structured models by simulations. For each model, one samples independent realisations of the interaction matrix $A$ and sets the growth rates vector to the mean structural vector $v$ of the model. The equilibrium is found according to Equation~\eqref{eq:def:x*} and $P_S$ is estimated by the proportion at which the equilibrium is admissible. The simulation results are analogous to the random model in the sense that $P_S$ decreases exponentially towards zero with $S$. Note that we also performed simulations for the random model, whose results correspond nicely to the analytical prediction (see Fig.~\ref{fig:sim:proba:persist} (a)).

\subsubsection{Analytical results}\label{sec:analytical:moderate}
First, we express the law of $x_i^{*}$ for arbitrary large $S$. Our results ensue from the direct application, with some extensions, of Geman's work~\cite{Geman1982} on solution of random large systems. They hold under the following assumptions:
\begin{assump}\label{assum:delta1/2}
In the Lotka-Volterra model~\eqref{LVWM}, we assume that
\begin{enumerate}[(i)]
\item $\delta=\frac{1}{2}$;
\item the interaction matrix $A=(a_{ij})$ has i.i.d.~entries with common mean $\E(a_{11})=0$; 
\item the intrinsic growth rates vector $r = (r_i)$ has i.i.d.~entries;
\item $r$ and $A$ are independent;
\item the second moments of the laws of $A$ and $r$ satisfy $\frac{\E(r_1^2)\E(a_{11}^2)}{\vert\theta\vert^2}<\frac{1}{4}$;
\item there exists a constant $\kappa$ such that $\E(\vert a_{11}\vert^S)< S^{\kappa S}$, for all $S\ge 2$;
\item the matrix $(\theta {I} +A/\sqrt{CS})$, where ${I}$ denotes the $S\times S$ identity matrix, is nonsingular.
\end{enumerate}
\end{assump}
The assumptions (v) and (vi) on the second and higher moments of $A$ and $r$ are very natural and not restrictive, since they are satisfied by a wide collection of laws (normal, uniform, beta, gamma, lognormal,...) as long as the chosen standard deviation is not too large. 
Note also that the mean structural vector satisfies these assumptions.
The last condition allows us to consider that the solution to~\eqref{eq:def:x*} always exists.

\begin{prop}\label{prop:asympt:solut:sublinear}
Under Assumptions~\ref{assum:delta1/2}, the biomasses at equilibrium of the model~\eqref{LVWM} converge in law towards Gaussian random variables:
\begin{equation}\label{EquilibriaNormal}
x^{*}_i \xRightarrow[ ]{\mathcal{L}} \mathcal{N}\lp-\frac{\E\lp r_1\rp}{\theta},\ \frac{\Var(r_1)}{\theta^2}+\frac{\E\lp r_1^2\rp\sigma^2}{\theta^2\lp\theta^2-\sigma^2 \rp}\rp,\quad\text{ for all } i=1,\ldots, S.
\end{equation}
Moreover, for every fixed $1\leq k \leq S$, the collection $\left(x^{*}_1,...,x^{*}_k \right)$ has independent entries.
\end{prop}

\proof
The proof is a generalisation of a result given in \cite{Geman1982}. We will assume without loss of generality that $C=1$. Indeed, $\E(\frac{1}{\sqrt C}a_{ij})=0$ and $\Var(\frac{1}{\sqrt C}a_{ij}) = \Var(a_{ij}|a_{ij}\neq 0)$ so that in the case of moderate interactions the variance of the $a_{ij}$ is preserved when dividing any $a_{ij}$ by $\sqrt C$.  Consider the system (\ref{LVWM}),
$$-\theta x^{*} = r + \frac{1}{\sqrt S} A x^{*}.$$
First approach the solution of the system by a Neumann progression which is given by
$$-\theta x^{*} = r + \sum_{k=1}^{\infty}\lp\frac{A}{\sqrt S} \rp^k \lp-\frac{1}{\theta} \rp^k r. $$
As in \cite{Geman1982}, define
\begin{equation}\label{eq:def:alpha}
\alpha_{r,\theta}\lp i,k,S \rp = \lc \lp\frac{A}{\sqrt S} \rp^k \lp-\frac{1}{\theta} \rp^k r \rc_i = \lp-\frac{1}{\theta} \rp^k \lp \frac{1}{\sqrt S} \rp^k \sum_{l_1,...,l_k} a_{il_1}a_{l_1l_2}...a_{l_{k-1}l_k}r_{l_k},
\end{equation}
which is simply the $i$th component of $\lp\frac{A}{\sqrt S} \rp^k \lp-\frac{1}{\theta} \rp^k r$. %The idea will be to 
We compute the joint moments of the $\alpha_{r,\theta}$ and show that they are the same as those of normal random variables. For $m$ fixed and distinct pairs $\lp k_j, n_j \rp_{1\leq j \leq m}$ we thus compute $\E\lp \prod_{j=1}^m \alpha_{r,\theta}\lp i_j, k_j, S \rp^{n_j} \rp$. This involves the same combinatorics as the one employed in \cite{Geman1982} as mixing terms from \eqref{eq:def:alpha} arise. Indeed for asymptotical contributions, each chain $a_{l_1l_2}...a_{l_{k-1}l_k}r_{l_k}$ has to be paired exactly with itself. Using moreover independence of the $(a_{ij})$ and the $(r_i)$, this gives the joint moments
$$E\lp \prod_{j=1}^m \alpha_{r,\theta}\lp i_j, k_j, S \rp^{n_j} \rp = 
\begin{cases}
\prod_{j=1}^m \E(r_1^2)^{\frac{n_j}{2}} \lp\frac{\sigma}{\theta}\rp^{k_j n_j} \prod_{p=1}^{n_j / 2} \lp2p -1 \rp & \text{ if every $n_j$ are even}\\
0 & \text{ otherwise.}
\end{cases}$$
These are the joint moments of independent Gaussian random variables, i.e.
$$\lp \alpha_{r,\theta}\lp i_1, k_1, S \rp, ..., \alpha_{r,\theta}\lp i_m, k_m, S  \rp \rp \xRightarrow[ ]{\mathcal{L}} \lp Z_1, ..., Z_m \rp, $$
where $Z_j \sim \mathcal{N}\lp0,\; \E\lp r_1^2\rp \frac{\sigma^{2k_j}}{\theta^{2k_j}} \rp$ are independent.
Returning now to the Neumann progression and by calculating variance and expectation, we find that when $S \rightarrow \infty$
%$$- \theta x_i^{S,p}  \Rightarrow \mathcal{N} \lp \E(r_1), \Var\lp r_1 \rp + \E\lp r_1^2 \rp \frac{\frac{\sigma^2}{\theta^2}- \frac{\sigma^{2p}}{\theta^{2p}}}{1- \frac{\sigma^2}{\theta^2}} \rp.$$
%One can now show that for $p \rightarrow \infty$,
$$\lim_{S\to \infty}x_i^* = x_i^{*,\infty} \sim \mathcal{N}\lp -\frac{\E\lp r_1 \rp}{\theta}, \ \frac{\Var(r_1)}{\theta^2} +  \frac{\E\lp r_1^2 \rp\sigma^2}{\theta^2\lp\theta^2 - \sigma^2 \rp} \rp.$$
\endproof

With this result, one shows that asymptotically, there exists almost surely no feasible equilibrium for the system (\ref{LVWM}).
\begin{thm}\label{thm:feasibility:delta1/2}
Under Assumption~\ref{assum:delta1/2}, the probability that an equilibrium of model~\eqref{LVWM} is feasible tends toward zero. That is 
$\lim_{S \to \infty}P_{S}  =0 $.
\end{thm}
\begin{proof}
Let us define $P_{S}^{(k)}=\Prob\lp x^{*}_j>0,\, \forall j=1,\ldots,k\rp$. With this notation, $P_{S} = P_{S}^{(S)}$. The convergence in law in~\eqref{EquilibriaNormal}, as well as the independence of $(x_i^*)$ imply
\begin{eqnarray*}
0 \leq \lim_{S\to \infty} P_S & \leq & \limsup_{S\to \infty} P_{S}^{(S)}\\
 & \leq & \limsup_{S\to \infty} P_{S}^{(k)} \\
& = & \Prob\lp x^{*}_j>0,\, \forall j=1,\ldots,k \rp \\
& = & \Phi\lp\frac{\E(r_1)}{\sqrt{\Var(r_1) + \E\lp r_1^2 \rp \frac{\sigma^2}{\theta^2 - \sigma^2}}}   \rp^{k},
\end{eqnarray*}
for all $k\geq1$ and where $\Phi$ denotes the standard Gaussian cumulative distribution function.
Since the probability on the right-hand side is strictly less than 1, this implies that $\limsup_{S\to \infty} P_{S}=0$.
\end{proof}

When considering the mean structural vector of the random model $r=v$ (see Table~2), the previous proof allows to approximate the probability that an equilibrium is feasible in the following way:
\begin{equation}\label{eqn:ps:moderate:random}
P_S \asymp \Phi\lp \sqrt{\frac{\theta^2-\sigma^2}{\sigma^2}} \rp^S,
\end{equation}
which is represented in Fig.~\ref{fig:sim:proba:persist}.

\subsubsection{Simulations}\label{sec:simu}
The exponential decrease of $P_S$ is not restricted to the random model, for which our analytical result holds. Indeed, we have computed numerically the probability of feasibility for the other models, as well as for the random model (see Fig.~\ref{fig:main1}). The growth rates vector is set to the mean structural vector for each model respectively (Table~2).

An interaction is non-zero with probability $C$. We choose $a_{ij}\sim\mathcal{N}(0,\sigma)$ for any non-zero entries of the interaction matrix $A$ in the random model, i.e.~when the sign of the interaction does not matter. In this sense, $\E(a_{ij}) = 0$ and $\Var(a_{ij}) = C\cdot\Var(a_{ij} | a_{ij}\neq 0) = C\cdot\sigma^2$. In the other cases, a strictly positive interaction is randomly drawn from a folded normal distribution such that $a_{ij}\sim\left|\mathcal{N}(0,\sigma)\right|$. The expectation is $\E(a_{ij}| a_{ij}\neq 0) = \sigma\sqrt{2/\pi}$ and the variance $\Var(a_{ij} | a_{ij}\neq 0)=\sigma^2\lp  1- 2/\pi\rp$. A strictly negative interaction is similarly sampled such that $a_{ij}\sim-\left|\mathcal{N}(0,\sigma)\right|$.

In Fig.~\ref{fig:main1}, the connectance as been fixed to $C=0.25$ for each model. To reach such a connectance in the case of the niche model, we choose the niche values uniformly on the interval $\left[0 ; 1 \right]$, and their breadth according to a Beta random variable with shape parameters $(1,3)$.

%%%%%%%%%%%%%%%%%%%%%%%%%%%%%%%%

\begin{figure}[H]
\centering
\includegraphics{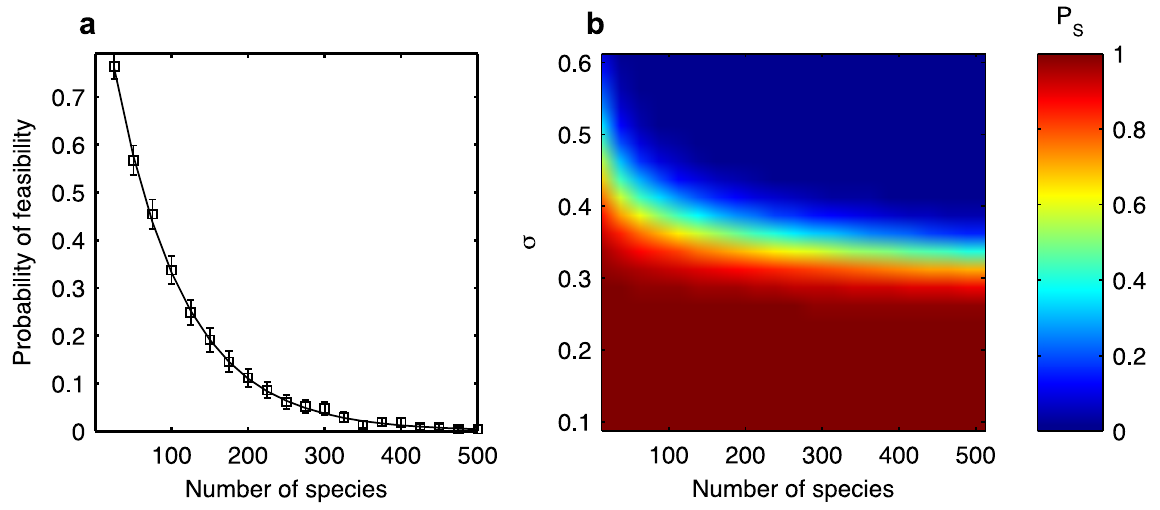}
\caption{Probability of feasibility as a function of number of species $S$ and standard deviation of interaction strengths $\sigma$ in the random model under moderate interactions and with the mean structural vector. (\textbf{a}) The continuous line denotes the analytical predictions of $P_S$, while the error bars are $95\%$ confidence intervals from Monte-Carlo estimates for $1000$ simulations. The parameters are $C=1$, $\sigma=0.4$, $\theta=-1$. (\textbf{b}) Analytical prediction of the probability of feasibility with respect to $\sigma$ and $S$.}\label{fig:sim:proba:persist}
\end{figure}

In the case of the niche and nested-hierarchy models, we also simulated the equilibria feasibility when not using the mean structural vector $v$. In Fig.~\ref{fig:nicheNested:rRandom:moderate}, the growth rates are i.i.d.~gaussian random variables with standard deviation fixed to $0.15$ and mean one (in Fig.~\ref{fig:nicheNested:rRandom:moderate} (a-d)) or mean zero (in~Fig.~\ref{fig:nicheNested:rRandom:moderate} (e-f)). In the former, the probability of feasibility rapidly decreases towards zero, similarly as what is predicted with random models. When $\E(r_i) = 0$ for all $i$, the equilibria abundances are centered around zero with a positive variance independently of the species index $i$. As such, there is no chance to observe $P_S > 0$ for this choice of $r_i$.

\begin{figure}[h!]
\centering
\includegraphics{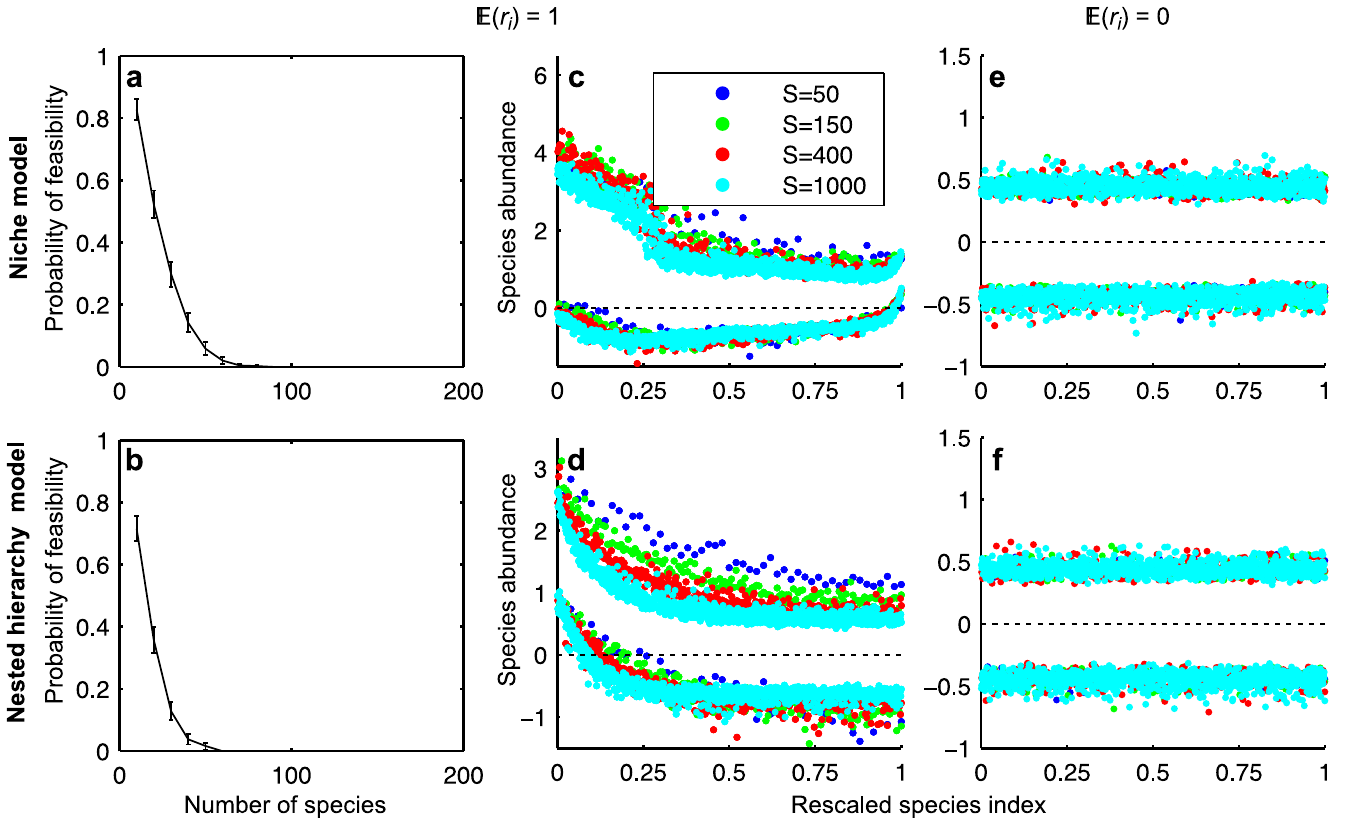}
\caption{Equilibria simulated from the niche and nested-hierarchy models under moderate interactions and for i.i.d.~random growth rates. (\textbf{a}-\textbf{b}) The probability of feasibility decreases rapidly when $E(r_i) = 1$. Envelop (the maximum and the minimum value) of the equilibria for (\textbf{c}-\textbf{d}) $\E(r_i) = 1$ and (\textbf{e}-\textbf{f}) $\E(r_i) = 0$. In this case $P_S$ is always zero. The species have been assigned a number between zero (top species) and one (basal species) corresponding to their hierarchy in the web. $500$ simulations have been performed for each $S$. The growth rates are normally distributed with a fixed standard deviation $\sqrt{\textrm{Var}(r_i)}=0.15$.}\label{fig:nicheNested:rRandom:moderate}
\end{figure}

%%%%%%%%%%%%%%%%%%%%%%%%%%%%%%%% 

%\begin{figure}
%\centering
%\begin{subfigure}[b]{0.5\textwidth}
%\includegraphics[width=\textwidth]{../codes/janvier2015/deltaUnDomFeasability/cascadeGrowthRates.pdf}
%\caption{Growth rates}
%\end{subfigure}~\hspace{3ex}~
%\begin{subfigure}[b]{0.5\textwidth}
%\includegraphics[width=\textwidth]{../codes/janvier2015/deltaUnDomFeasability/cascadeEquilibria.pdf}
%\caption{Abundances at equilibrium}
%\end{subfigure}~\hspace{3ex}~
%\caption{Abundances at equilibrium for the cascade model depending on the growth rates vector $r$. One realisation with $N=500$, $\sigma=0.4$ (both for positive and negative entries of $A$), $\mu_{A_+} = 2$, $\mu_{A_-} = -2.8$ and $\theta = -1$ is represented. (a) The growth rates vectors of each species $i$. (b) The corresponding abundances at equilibrium. The mean structural vector drives each abundance to one and enables thus feasibility. This is not the case for other choices of $r$. Note that because of the definition of the cascade model, the species are labeled from $i=1$ (top predator) to $i=500$ (basal species).
%}\label{fig:delta1:cascade}
%\end{figure}

%%%%%%%%%%%%%%%%%%%%%%%%%%%%%%%%

\subsection{Weak interactions}
For the case $\delta=1$, we provide the conditions under which the system~\eqref{LVWM} possesses almost surely a feasible equilibrium. In this situation, the interactions between the species become extremely weak when the size of the system grows. Compared to the moderate case, the variance of the solution approaches zero allowing us to use a law of large numbers. This drives each component $x_i^*$ of the solution to a constant proportional to its intrinsic growth rate $r_i$ as illustrated in~Fig.~\ref{fig:delta1:traj}. This leads to feasible equilibria for any positive growth rate. 

Analytical results are given in the case of unstructured models, whereas structured models are explored by mean of simulations.

\subsubsection{Analytical results}\label{sec:analytical:weak}
We begin by giving the assumptions under which our analytical results hold :

\begin{assump}\label{assum:delta1}
In the Lotka-Volterra model~\eqref{LVWM}, we assume that
\begin{enumerate}[(i)]
\item $\delta=1$;
\item the interaction matrix $A=(a_{ij})$ has i.i.d.~entries with common mean $\E(a_{11})=C\mu_{A}$ so that $|\mu_{A}|<|\theta|$; 
\item the intrinsic growth rates vector $r = (r_i)$ has bounded entries;
\item the law of the $(a_{ij})$ satisfies $\E(a_{ij}^8)<\infty$;
\item the matrix $(\theta {I} +A / (CS))$ is nonsingular.
\end{enumerate}
\end{assump}
As for the Assumptions~\ref{assum:delta1/2}, these conditions are not restrictive on the choice of the distributions for the $(a_{ij})$. The mean structural vector satisfies the Assumption~\ref{assum:delta1} (iii) for all the models that are considered. In (ii), $\E(a_{11})=C\mu_{A}$ is equivalent to $\E(a_{11}\mid a_{11}\neq0) = \mu_{A}$ since $C$ is the probability that an entry of the matrix is set to $0$. Note furthermore that, in the settings of weak interactions, the entries of $A$ are divided by $CS$ so that $\E(a_{ij}/(CS))=\mu_{A}/S$.

\medskip

We first introduce the analytical result for the random model.

\begin{thm}\label{thm:feasibility:delta1}
Under Assumption~\ref{assum:delta1} with $\mu_{A}=0$ the asymptotic equilibrium of the model~\eqref{LVWM} is feasible and is given by
$$x_i^{*,\infty}  = \lim_{S \to \infty}x_i^{*} =  \frac{r_i}{\vert\theta\vert}  \text{ almost surely}, \quad \text{ for all } i=1,\ldots,S.$$
Moreover, $x_i^{*,\infty}$ is almost surely feasible for any $r$ such that $r_i>0$ for all $i$.
\end{thm}

\begin{proof}
%In the following we will assume without loss of generality that $C=1$. Indeed $\E(a_{ij}) = C\cdot\E(a_{ij}|a_{ij}\neq0)$ so that under weak interactions $\E(a_{ij})$ is preserved when dividing any $a_{ij}$ by the connectance $C$. \todo{Here, I think we don't have to explain this things with the connectance, as there are no computations/formula in the proof.}
The convergence of $x_i^{*}$ toward $ \frac{r_i}{\vert\theta\vert}$ follows  directly from results on the solution of large random systems of linear equations, \cite[Thm.~1]{Geman1982}. 
\end{proof}
In the case of the random model, setting the growth rates vector to $v=|\theta|\cdot \mathbbm{1}$ leads thus to $x_i^{*,\infty} = 1$ for all $i$. But Theorem~\ref{thm:feasibility:delta1} is more general and allows to consider any growth rates vector $r$ with bounded entries. In Fig.~\ref{fig:delta1:traj} for example, the entries of $r$ are chosen to take only two values and one sees that the equilibrium vector also converges only toward two values (proportional to $r_i$) when $S$ becomes large.

In Assumption~\ref{assum:delta1}, the expectation $\mu_{A}$ does not need to be zero. This enables the derivation of analytical results in the case of random models in which the mean interaction is positive (mutualistic and commensalistic interactions) or negative (competitive and amensalistic interactions). Indeed, a generalisation of the previous proof by allowing arbitrary $\mu_{A}$ leads to the following:
\begin{thm}\label{thm:feasibility:delta1:b}
\label{equridiff}
Under Assumption~\ref{assum:delta1}, the asymptotic equilibrium of the model~\eqref{LVWM} is given by
$$x_i^{*,\infty}  = \lim_{S \to \infty}x_i^{*} = \frac{r_i}{|\theta|} + \frac{\mu_{A}}{|\theta|\lp |\theta| - \mu_{A} \rp} \bar r \quad \text{ almost surely},  \text{ for all } i=1,\ldots,S,$$
where $\bar r$ denotes the arithmetic mean of the entries of $r$. Moreover, if $r_i \geq \frac{\mu_{A}}{\mu_{A}+\theta}\cdot \bar r$ for all $i=1,\ldots,S$, then $x^{*,\infty}$ is feasible with probability one.
\end{thm}
\begin{proof}
The proof is analog to the one provided in~\cite[Thm.~2]{Geman1982}, where we write $\hat W = W - M$ with $M$ the $S\times S$ matrix with every component set to $C\mu_{A}$, $W=A$ and $\hat W$ is an $S\times S$ random matrix with i.i.d.~centred entries.
\end{proof}

However, the competition and mutualistic models introduced in section~\ref{sec:interactionsModels} do not completely satisfy the Assumptions~\ref{assum:delta1}, since some dependences are introduced among the entries of the interaction matrix $A$. Indeed,
$a_{ij}=0 \Leftrightarrow a_{ji}=0,$
so that commensalistic/amensalistic interactions are forbidden. Moreover, the particular case of predation leads to a sign antisymmetric matrix $A$, i.e.~$a_{ij}>0 \Leftrightarrow a_{ji}<0$. In the following, we show that the same type of results are obtained for these cases. 

For mutualism or competition, we define the entries of $A$ in the following way,
$$a_{ij} = w_{ij}\cdot b_{ij} \quad \text{and} \quad a_{ji} = w_{ji}\cdot b_{ij},$$
where $(w_{ij})$ are i.i.d.~random variables of mean $\mu_{A}$ and variance $\sigma^2$ (e.g.~folded normal random variables) and $(b_{ij})$ are i.i.d.~Bernoulli random variables of parameter $C$. 

Note that for every $i,j\in\mathcal{S}$, $\E(a_{ij}) = \E(a_{ji}) = C\mu_{A}$, so that if we define $M$ to be the matrix with all elements set to $C\mu_{A}$, the matrix 
\begin{equation}\label{eq:hatW}
\hat W = A - M
\end{equation}
is centred and that $Cov(\hat w_{ij},\hat w_{ji}) = \mu_{A}^2\cdot C(1-C)\neq 0$. Like in the proof~\cite[Thm.~2]{Geman1982}, we need to show that $\|\hat W/S \| \to 0$. This assertion is based on~\cite{Geman1980}, which gives an upper bound for the norm of sample covariance matrix, and can be related to the work in~\cite{Furedi1981}, where the largest eigenvalue of random symmetric matrices is studied. Here we will show that $\|\hat W/S \| \to 0$ in the framework of~\cite{Credner2008}, where it is demonstrated that the limiting distribution of the eigenvalues of a sample covariance matrix $1/SVV^T$ remains the Mar\~cenko-Pastur law, even with some dependencies among the entries of a centred random matrix $V$. 

\begin{lemma}\label{lemma:Geman1980}
Consider the random matrix $\hat W$ defined by equation~\eqref{eq:hatW} and assume that there exists a constant $K$ such that $\E(|\hat w_{ij}|^k)\leq K^{2k}$ for any $k\leq S$ and $1\leq i,j \leq S$. Then
$$\|\hat W/S \| \to 0 \quad \text{almost surely}.$$
\end{lemma}
\begin{proof}
The conditions (MP1), (MP2) and (MP3) in~\cite{Credner2008} still hold for $\hat W$ in our framework. Thus letting $\lambda_{max}$ be the largest eigenvalue of $1/S\hat W \hat W^T$, the same combinatorial arguments can be used, and following~\cite[section 4]{Credner2008} we arrive to 
\begin{eqnarray*}
\E(\lambda_{max}^k) \leq \E\lp \textrm{Tr}(\frac{1}{S}\hat W \hat W^T)^k \rp &\leq & \frac{(2k)!}{k!(k+1)!}S + c\cdot K^{2k} + \mathcal{O}(1/S) \\
& < & 4^kS+c\cdot K^{2k} + \mathcal{O}(1/S),
\end{eqnarray*}
where $c$ is a constant and $\textrm{Tr}(\cdot)$ the trace operator. Letting $\epsilon>0$ and using the Markov inequality, we find
\begin{eqnarray*}
\Prob\lp \frac{1}{\sqrt S}\left\| \frac{1}{\sqrt S} \hat W \right\| > \epsilon \rp = \Prob\lp \lambda_{max}^k > S^k \epsilon^{2k} \rp \leq \frac{4^kS + cK^{2k}}{S^k\epsilon^{2k}} + \mathcal{O}(1/S^{k+1}),
\end{eqnarray*}
which goes to zero when $S$ is large and since $\sum_S \Prob\lp \frac{1}{\sqrt S}\left\| \frac{1}{\sqrt S} \hat W \right\| > \epsilon \rp < \infty$ for $k\geq3$, almost sure convergence holds.
\end{proof}

\begin{corollary}\label{cor:feasibility:delta1:mutPred}
Under Assumption~\ref{assum:delta1} with a matrix $A$ so that $a_{ij}=0 \Leftrightarrow a_{ji}=0$, and where $\E(a_{ij}|a_{ij}\neq0)=\mu_{A}$ and $r=\gamma\cdot\mathbbm{1}$, with $\gamma \in \R^*$, the asymptotic equilibrium of the model~\eqref{LVWM} is given by
$$x_i^{*,\infty}  = \lim_{S \to \infty}x_i^{*} = \frac{\gamma}{|\theta|-\mu_{A}} \text{ almost surely},$$
$\text{ for all } i=1,\ldots,S$. Under these assumptions and for $\gamma/(|\theta|-\mu_{A})>0$, $x^{*,\infty}$ is feasible with probability one.
\end{corollary}
\begin{proof}
Since each line in $A$, i.e.~the collection $\left(a_{ij}\right)_{1\leq j\leq S}$ for arbitrary $i\in\mathcal{S}$, still contains independent elements for arbitrary $S$, the proof is analog to the one of Thm~\ref{thm:feasibility:delta1:b} by using Lemma~\ref{lemma:Geman1980}.
\end{proof}

In the case of predation, interactions have the form $(+,-)$, so that the previous Corollary extends in the following way.

\begin{corollary}\label{cor:feasibility:delta1:pred}
Consider Assumption~\ref{assum:delta1} with a matrix $A$ so that $\mu_{A}=\frac{\mu_{A_+} + \mu_{A_-}}{2}$, where $\E(a_{ij}|a_{ij}>0)=\mu_{A_+}>0$, $\E(a_{ij}|a_{ij}<0)=\mu_{A_-}<0$, and with $a_{ij}=0 \Leftrightarrow a_{ji}=0$ and $\sign(a_{ij}) = -\sign(a_{ji})$. Then for $r=\gamma\cdot\mathbbm{1}$, with $\gamma \in \R^*$, the asymptotic equilibrium of the model~\eqref{LVWM} is given by
$$x_i^{*,\infty}  = \lim_{S \to \infty}x_i^{*} = \frac{\gamma}{|\theta|-\frac{\mu_{A_+} + \mu_{A_-}}{2}} \text{ almost surely},$$
$\text{ for all } i=1,\ldots,S$. Under these assumptions and for $\gamma/\lp|\theta|-\frac{\mu_{A_+} + \mu_{A_-}}{2}\rp>0$, $x^{*,\infty}$ is feasible with probability one.
\end{corollary}

Consequently, the mean structural vectors (see Table~2) leads almost surely $x^*$ to the vector $\mathbbm{1}$ for every unstructured model.

\subsubsection{Simulations}\label{sec:simu:weak}
\textbf{Unstructured networks.} As in the moderate case, simulations have been performed to illustrate our analytical results. We illustrate the results for the random model in Fig.~\ref{fig:delta1:traj}. The outcomes for the other unstructured models are analogous. All random variables $a_{ij}$ are defined as in section~\ref{sec:simu}, with $C=0.25$, $\sigma=0.4$ and $\theta=-1$.

\begin{figure}[h]
\centering
\includegraphics{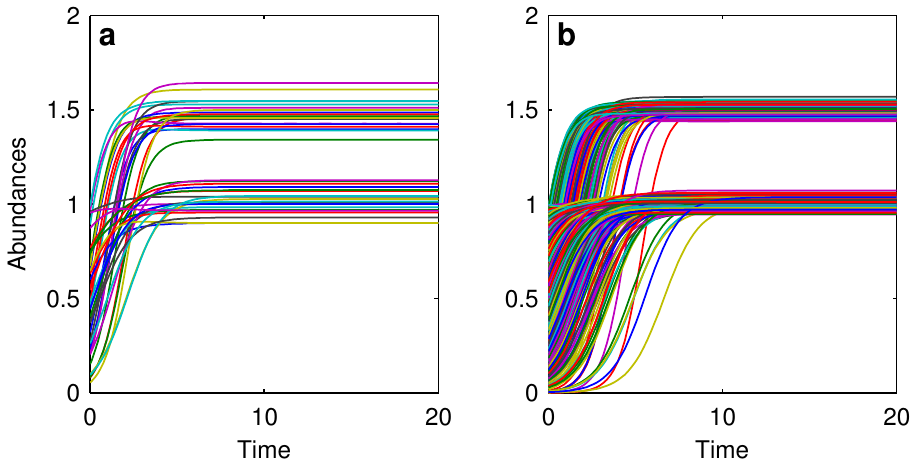}
\caption{Trajectories of the abundances of species through time in the case of weak interactions ($\delta = 1$) for the random model. Half of the species have a growth rate set to $1$, the other half set to $1.5$. As predicted analytically (Thm~\ref{thm:feasibility:delta1}), the abundances stabilise around their intrinsic growth rates ($\vert\theta\vert = 1$ and $\sigma = 0.4$). The simulations have been performed with MATLAB ode45. (\textbf{a}) With 36 species. (\textbf{b}) With 500 species. Note that for $S\to\infty$, equilibrium abundances converge to $1$ and $1.5$.}\label{fig:delta1:traj}
\end{figure}

\noindent\textbf{Structured networks.} The case of the cascade model~\cite{Cohen1990community} is illustrated in Fig.~\ref{fig:delta1:structured} (a) and (b). We represent the envelop of the equilibria for different $S$ from 1000 simulations with $C=0.25$, $\sigma=0.4$ and $\theta=-1$. The standard deviation of any $x^*_i$ decreases towards zero when $S$ increases, independently of the role of species $i$. Using the mean structural vector, the same convergence to the equilibrium $\mathbbm{1}$ as in the random model is hence obtained. We hypothesise that this result is a consequence of the cascade model yielding adjacency matrices with inferior triangular parts constructed like a Erd\H{o}s-R\'enyi network. The distribution of a top, an intermediate and a basal species are illustrated in Fig.~\ref{fig:delta1:structured:histogramNestedNiche} (a) and shows that any species abundance is independent on its role in the web, exactly as in random models.

For the niche model~\cite{Williams2000simple} and the nested-hierarchy model~\cite{Cattin2004}, a particular phenomenon occurs that is not observed for the other models. Indeed, simulations show that the equilibrium remains random when $S$ increases. This result is likely attributable to the construction of the models, with the standard deviation of the abundances depending on the hierarchical position of the species. Importantly, standard deviations do not decrease towards zero as the size of the network grows, as illustrated in Fig.~\ref{fig:delta1:structured} (d) and (f). Indeed, these highly structured networks and their related mean structural growth rates induce correlations among the abundances at equilibrium that prevent a convergence to a deterministic value, but rather takes place on a compact support. This is illustrated by the envelop of the equilibria that is represented in Fig.~\ref{fig:delta1:structured} (c) and (e) for different values of $S$. The empirical distribution of $x_i^*$ for particular species roles $i$ is illustrated in Fig.~\ref{fig:delta1:structured:histogramNestedNiche}. From these results, it is apparent that basal species may rapidly converge to a deterministic constant. However, for intermediate and top species, simulations show that the convergence is very likely to occur on a non-trivial support.

\begin{figure}[h!]
\centering
\includegraphics{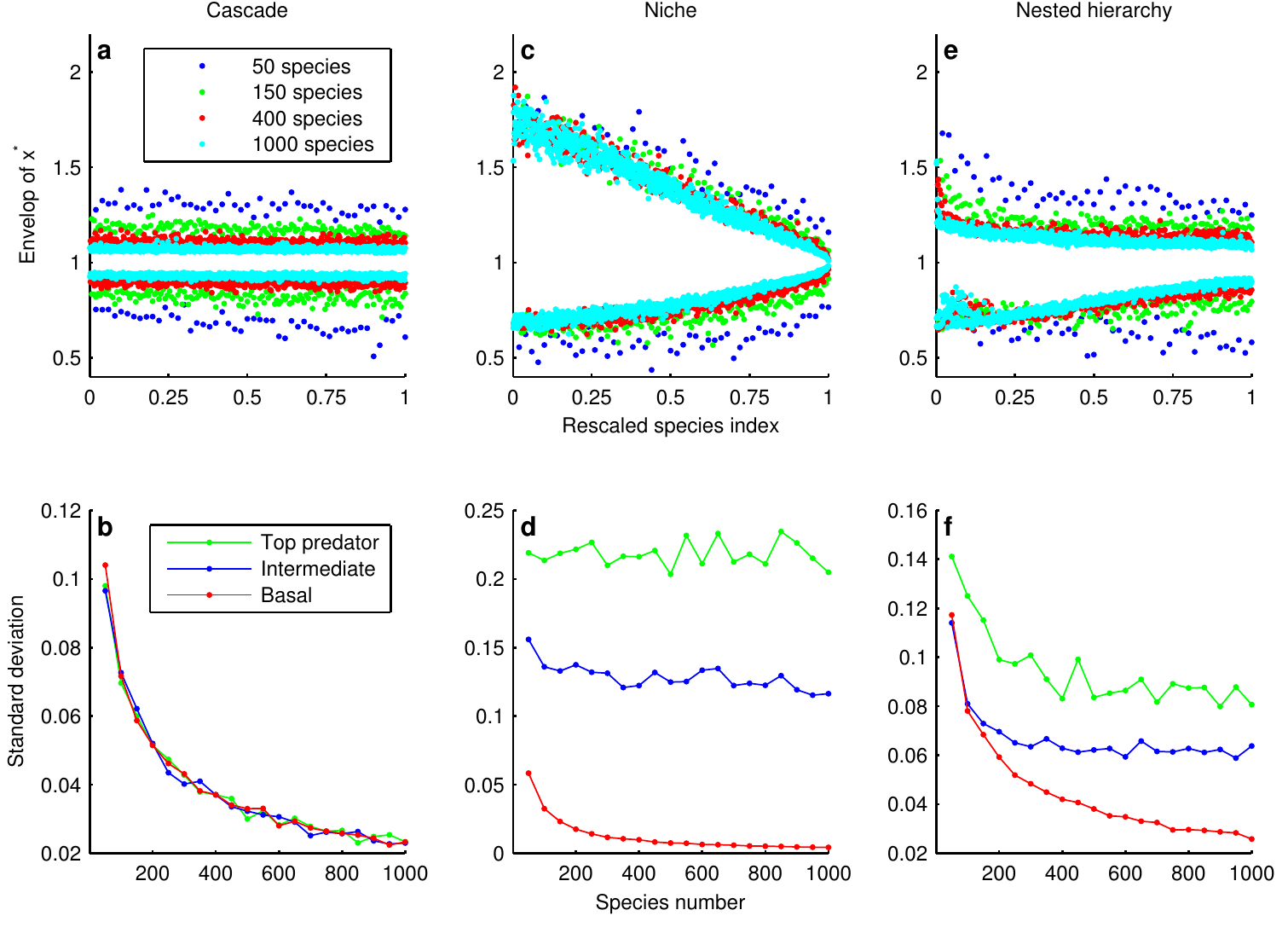}
\caption{Behaviour of the structured models under weak interactions with the mean structural vector. In the first row, we plot for different values of $S$ the envelop (the maximum and the minimal value) of the equilibrium $x^*$ among 1000 simulations. The species have been assigned a number between zero (top species) and one (basal species) corresponding to their hierarchy in the web. In the second row, the standard deviations of $x_{i}^{*}$ for $i=1$ (top predator), $i=S/2$ (intermediate) and $i=S$ (basal) are represented as a function of $S$. The parameters are $C=0.25$, $\sigma=0.4$ and $\theta=-1$. (\textbf{a} and \textbf{b}) In the cascade model, although hierarchically ordered, species tend to behave similarly with regard to convergence as $S$ grows. The equilibrium converges almost surely to the vector $\mathbbm{1}$. (\textbf{c} - \textbf{f}) The niche and the nested-hierarchy models keep the randomness of the equilibrium when $S$ grows, even under weak interactions.}\label{fig:delta1:structured}
\end{figure}

\begin{figure}[h!]
\centering
\includegraphics{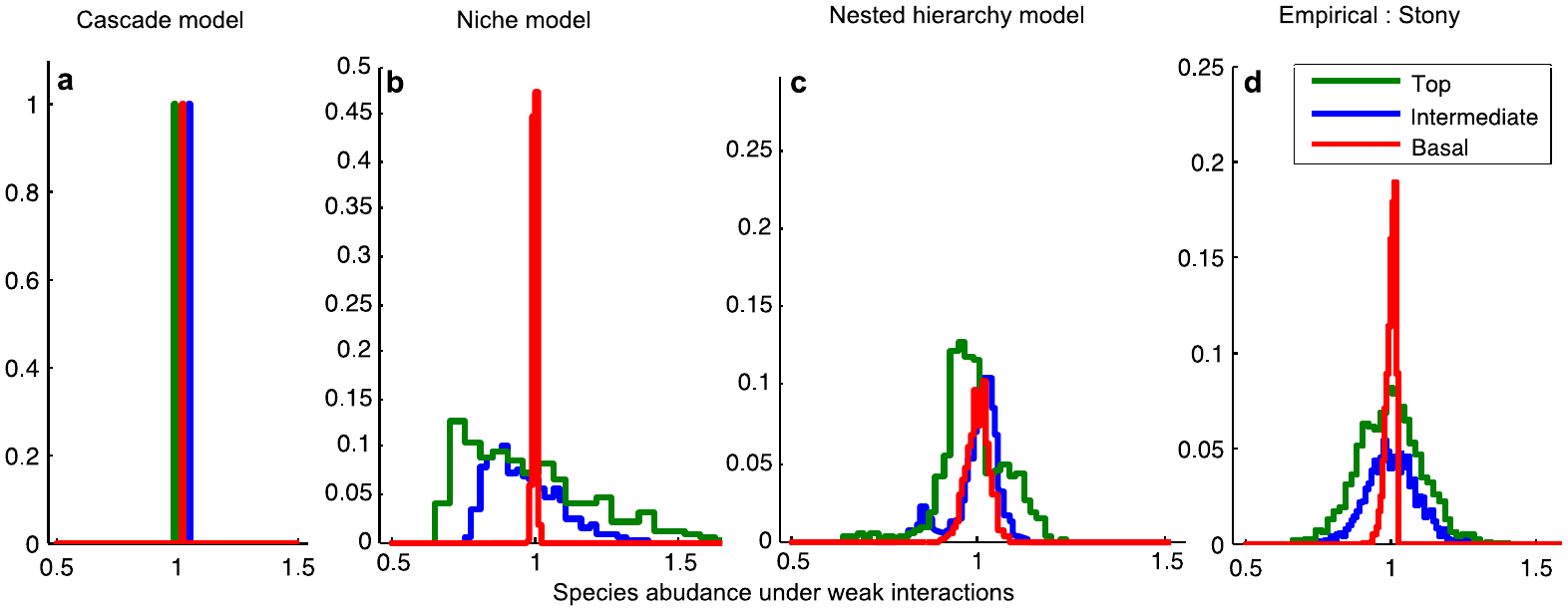}
\caption{Empirical distribution of top, intermediate and basal species under a regime of weak interactions. (\textbf{a}) Cascade model with $S=800$. (\textbf{b}) Niche model with $S=800$. (\textbf{c}) Nested-hierarchy model with $S=800$. The same parameters as in Fig.~\ref{fig:delta1:structured} have been used among 1000 simulations. (\textbf{d}) The distribution of a top, intermediate and basal species from the empirical food web Stony is illustrated. The interactions strength have been simulated similarly as for any structured web (see section~\ref{sec:interactionsModels}). The mean structural vector has been simulated with Monte-Carlo methods (200 trials).}\label{fig:delta1:structured:histogramNestedNiche}
\end{figure}

We also tested i.i.d.~Gaussian random growth rates with $\sigma_r = \sqrt{\Var(r_i)}=0.15$ for the niche and the nested-hierarchy models in Fig.~\ref{fig:nicheNested:rRandom:weak}. When $\E(r_i)=1$, a similar phenomenon as what was observed with the mean structural vector appears. The standard deviations of $x_i^*$ converge towards positive numbers and the resulting distribution of $x_i^*$ depends on the species index $i$. However, when $\E(r_i) = 0$ or when the standard deviation of $r_i$ increases, this particular hierarchical behavior disappears, as illustrated in Figs.~\ref{fig:nicheNested:rRandom:weak} and~\ref{fig:nicheNested:rRandom:weak:variances}. Every equilibria $x_i^*$ are similarly distributed around zero.

\begin{figure}[h!]
\centering
\includegraphics{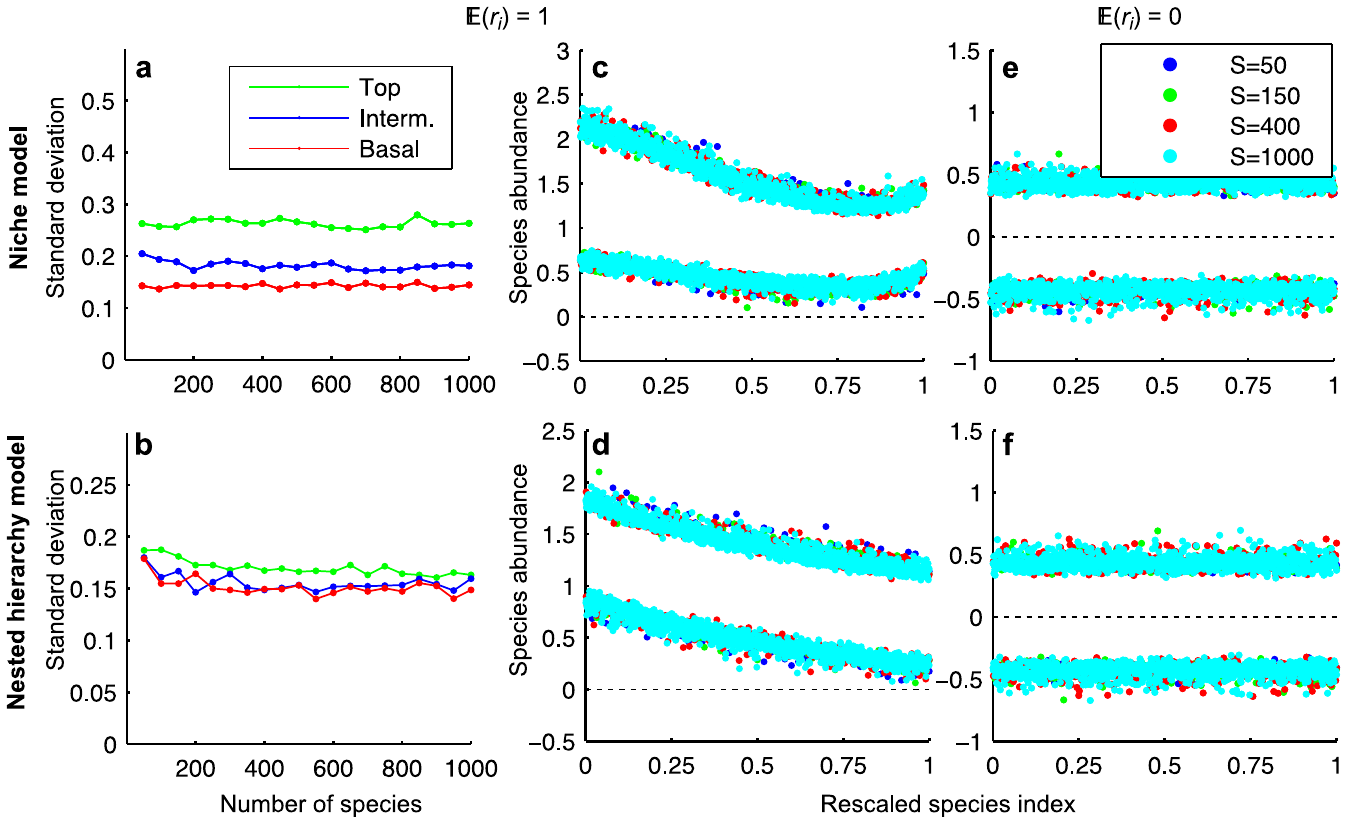}
\caption{Equilibria simulated from the niche and nested-hierarchy models under weak interactions and for i.i.d.~random growth rates. (\textbf{a}-\textbf{b}) The standard deviations of a top, an intermediate and a basal species are illustrated as a function of $S$ for $E(r_i) = 1$. Envelop of the equilibria for (\textbf{c}-\textbf{d}) $\E(r_i) = 1$ and (\textbf{e}-\textbf{f}) $\E(r_i) = 0$. $500$ simulations have been performed for each $S$. The growth rates are normally distributed with a fixed standard deviation $\sigma_r=0.15$.}\label{fig:nicheNested:rRandom:weak}
\end{figure}

\begin{figure}[h!]
\centering
\includegraphics{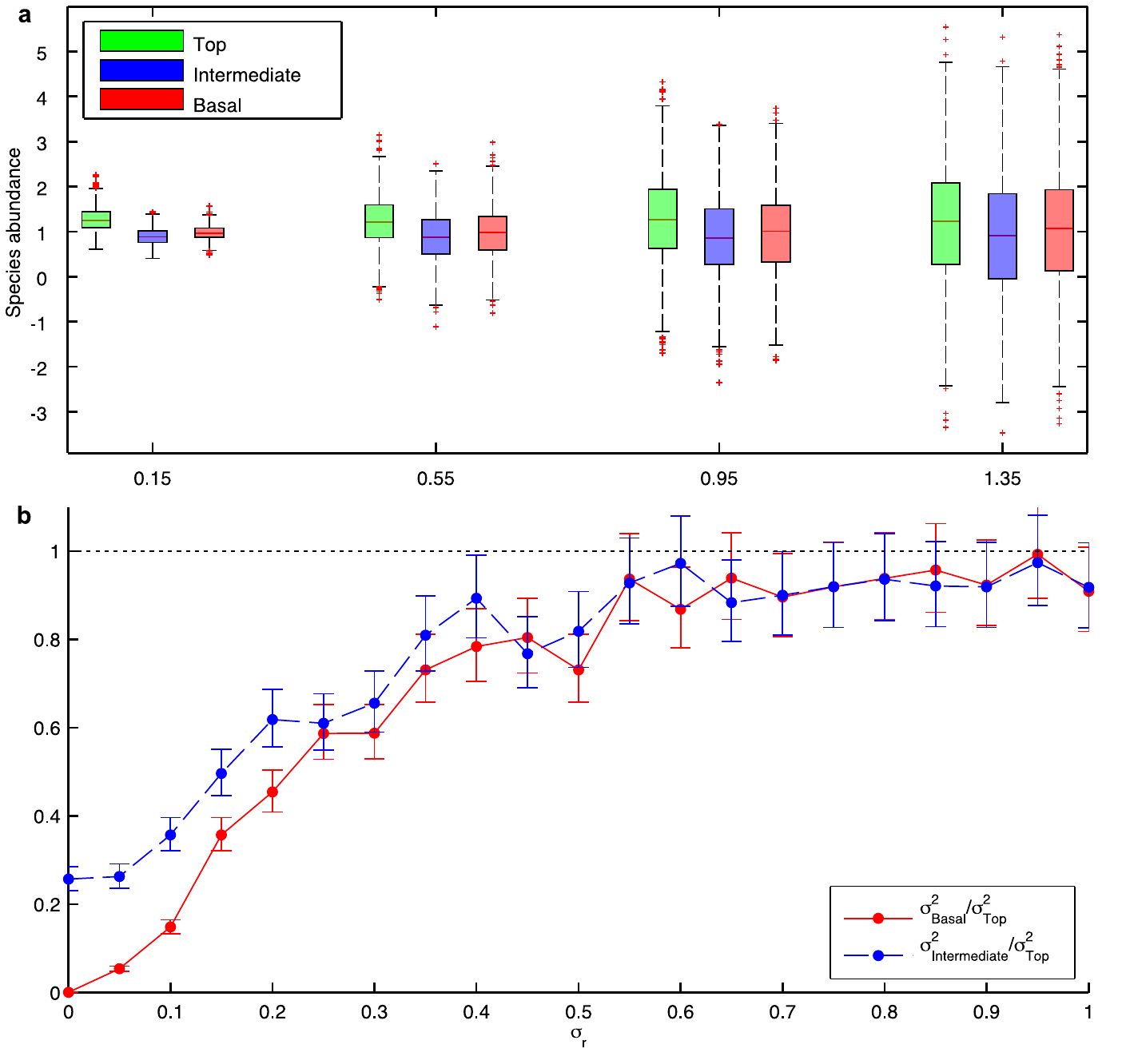}
\caption{When the standard deviation $\sigma_r$ of i.i.d.~random growth rates increases, equilibria variances of top, intermediate and basal species tend to become similar under weak interactions in a structured food-web (illustration for the niche model).  (\textbf{a}) Boxplot representation of the distributions. (\textbf{b}) Illustration of the variances ratio between basal and top (red trace, $\sigma^2_{\text{Basal}}/\sigma^2_{\text{Top}}$) and intermediate and top (blue trace, $\sigma^2_{\text{Intermediate}}/\sigma^2_{\text{Top}}$) as a function of $\sigma_r$ with a 95\% confidence interval. $1000$ simulations have been performed for each $\sigma_r$, $S=800$ and the growth rates are normally distributed with $E(r_i) = 1$.}\label{fig:nicheNested:rRandom:weak:variances}
\end{figure}

\subsection{Empirical networks}
We used our method on empirical food webs to illustrate how species abundances at equilibrium depend on their intrinsic role in the web. The topology is thus fixed to the observed one and the interactions are modeled as for any random structured web (see the section~\ref{sec:interactionsModels}). We report the results of four arbitrary networks. In Fig.~\ref{fig:empirical:meanStructuralVector}, the mean structural vector is used. In Fig.~\ref{fig:empirical:randomGrowthVector} two different random vectors with i.i.d.~components are chosen. The distributions under weak interactions are also illustrated in Fig.~\ref{fig:delta1:structured:histogramNestedNiche}, to allow better comparison with the niche and the nested-hierarchy models.

\begin{figure}[h!]
\centering
\includegraphics{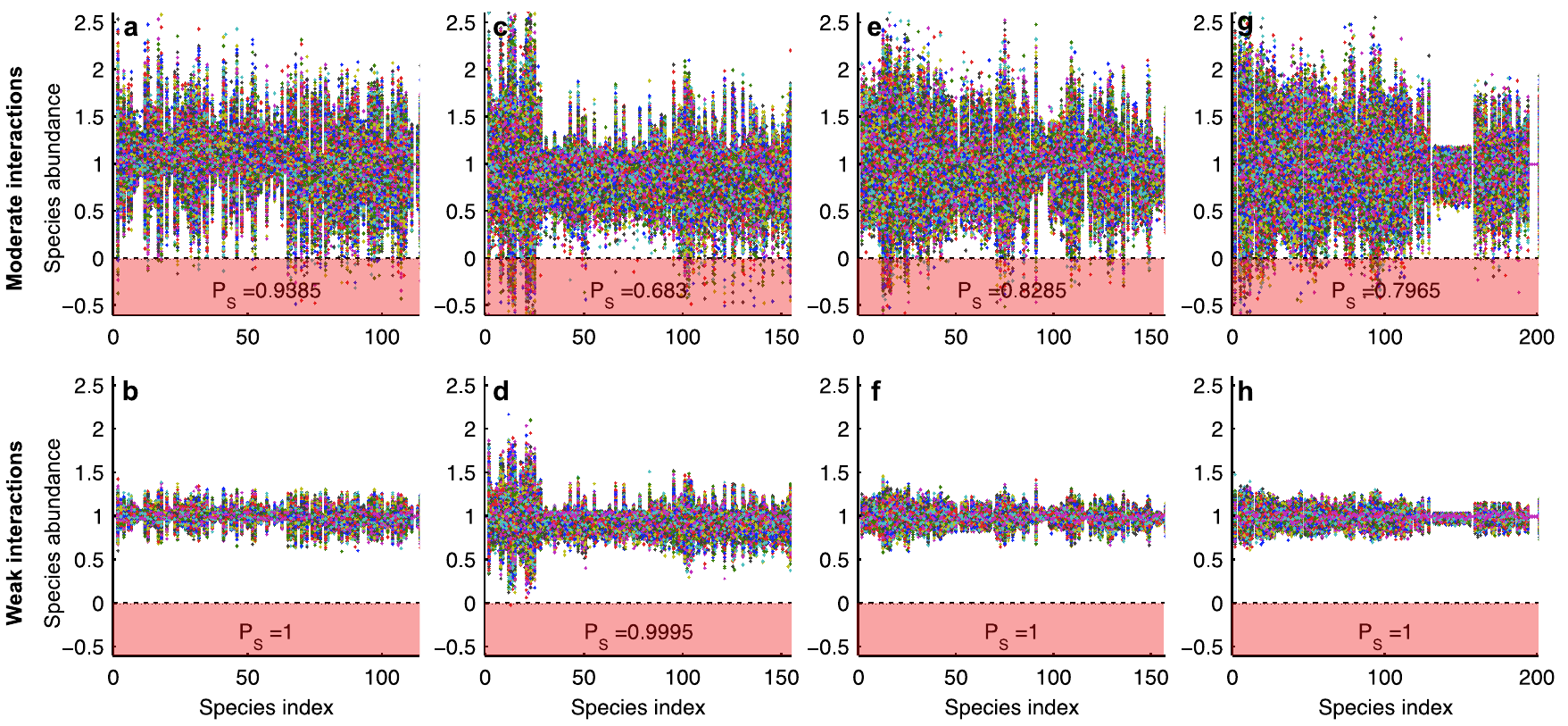}
\caption{Species abundance at equilibrium in four empirical food webs with the mean structural vector. (\textbf{a}-\textbf{b}) Stony, (\textbf{c}-\textbf{d}) Broom, (\textbf{e}-\textbf{f}) El-Verde and (\textbf{g}-\textbf{h}) Little Rock Lake. The mean of each $x_i^*$ is always one due to the mean structural vector, but the variance depends on the species index $i$, i.e.~on the species role in the web. 2000 simulations of $x^*$ are represented in each case. The mean structural vector has been simulated via Monte-Carlo methods (200 simulations).}\label{fig:empirical:meanStructuralVector}
\end{figure}

\begin{figure}[h!]
\centering
\includegraphics{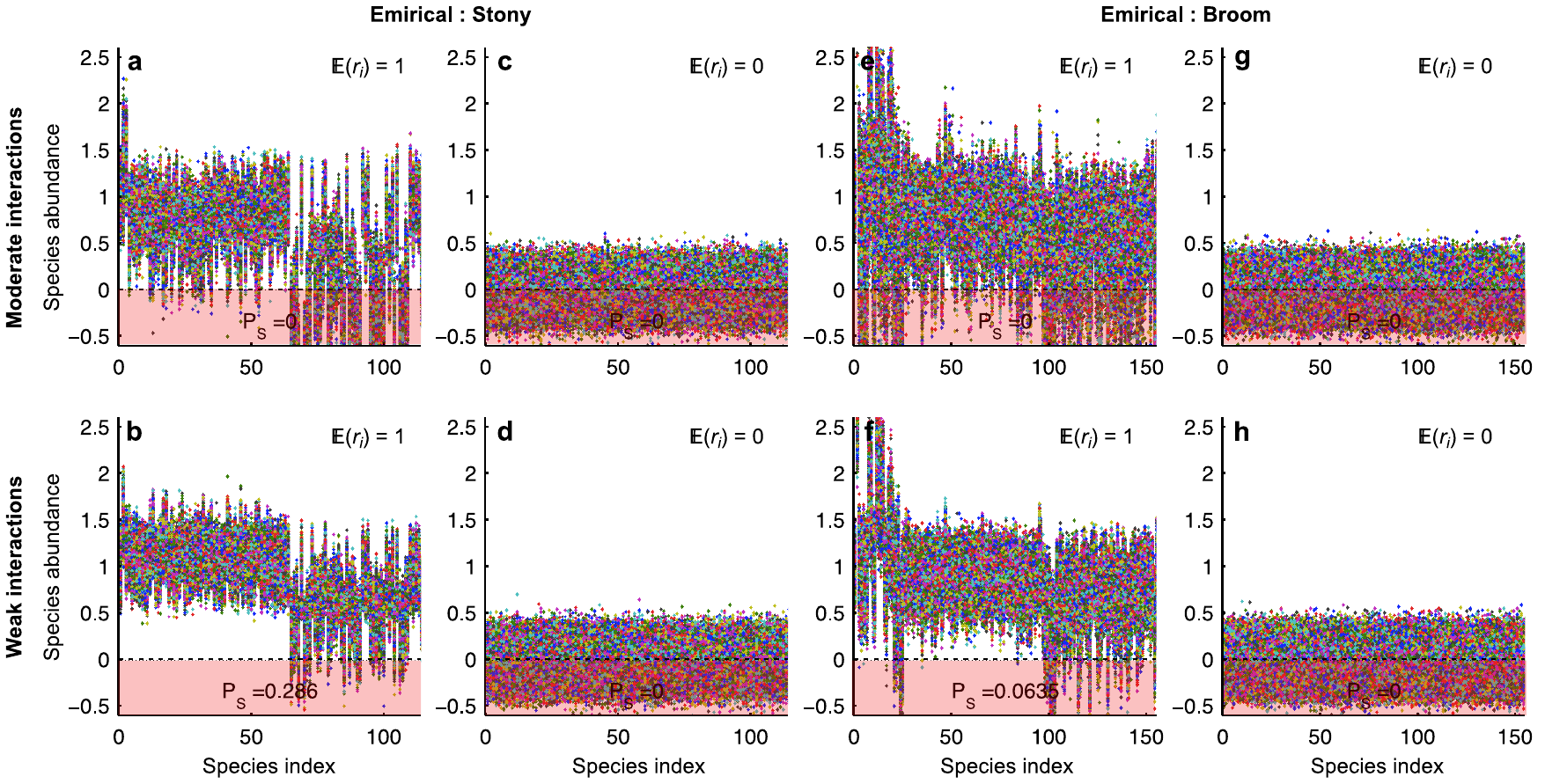}
\caption{Species abundance at equilibrium in two empirical food webs with random growth rates. (\textbf{a}-\textbf{d}) Stony and (\textbf{e}-\textbf{h}) Broom. The growth rates are i.i.d.~with mean one (\textbf{a}-\textbf{b} and \textbf{e}-\textbf{f}) or zero (\textbf{c}-\textbf{d} and \textbf{g}-\textbf{h}) and standard deviation fixed to 0.15. 2000 simulations of $x^*$ are represented in each case.}\label{fig:empirical:randomGrowthVector}
\end{figure}

\subsection{Connectance and interactions strengths}
In our approach, we fixed the connectance $C$ to arbitrary values in the mathematical developments, and to fixed values in the simulations. In this section, we briefly explore the consequence of a relationship between connectance and species number on feasibility, as observed in several contributions, e.g.~\cite{SchmidAraya2002,BanasekRichter2009}.
Observe first that, from the point of view of the average interaction strength, introducing zeros with a given probability $1-C$ in the interaction matrix is the same as multiplying the whole matrix by $C$. Therefore, one has
$$\frac{ A}{(CS)^{\delta}} = \frac{BC}{(CS)^{\delta}},$$ 
where $B$ is a matrix with all entries i.i.d.~and such that $\Prob(b_{ij}=0)=0$. Consider $C=\frac{1}{S^{\beta}}$, for $\beta\geq0$, which is a flexible function that adequately captures observed relationships between $C$ and $S$. We get
$$\frac{BC}{(CS)^{\delta}}=\frac{B}{S^{\delta+\beta(1-\delta)}}.$$
Above, we showed that feasibility is warranted when $\delta+\beta(1-\delta)>\frac{1}{2}$, i.e.~for
$$\beta>\frac{(\frac{1}{2}-\delta)}{(1-\delta)}.$$
Taking $\beta\geq\frac{1}{2}$ leads almost surely to feasible equilibria in the models considered here, independently of the exponent $\delta$. For $\beta<\frac{1}{2}$, $\delta$ can be smaller than 0.5. This shifts the three previously described regimes for $\delta$ to the left, so that $\delta$ can be smaller to reach the same results. Consequently, the existence of a relationship between $C$ and $S$ does not affect our conclusions.

%%%%%%%%%%%%%%%%%%%%%%%%%%%%%%%%%%%%%%%%%%%%%
\section{Consequences on local stability analysis}

Any local stability analysis should be preceded by a feasibility analysis~\cite{vandermeer1970,Roberts1974,Logofet1983Book}. Indeed, studying the local stability of a point that is not biologically feasible is not instructive on the behaviour of the network dynamics. Moreover, the Jacobian matrix of the system~\eqref{eq:def:J}, therefore its eigenvalues too, depend explicitly on $x^*$. Therefore, there is no warranty that randomly sampling directly the Jacobian may be sufficient to address local stability. Here, we explore this question, which extends the analyses of May and Allesina and co-workers~\cite{May1972,Allesina2012}.

\subsection{Random networks}

We first focus on the random model of May~\cite{May1972} in the case of moderate interactions in the Proposition~\ref{prop:stability:randomWeak} below. In this case, we find that May's criterion still holds when we randomly sample the interaction matrix $A$ (and not directly the Jacobian) under the additional condition that the equilibrium $x^*$ is feasible.

\subsubsection{Moderate interactions}
In Proposition~\ref{prop:stability:randomWeak}, we first show that May's criterion for stability, i.e.
$$\sqrt{CS}\tilde\sigma < |\theta| ,$$
is still sufficient for any feasible equilibrium when $\tilde\sigma$ is the standard deviation of the normalised interaction matrix $A/(CS)^{\delta}$, and not the standard deviation of the Jacobian matrix of the system. Recall that the case of moderate interactions is here very natural, as relying on Wigner's~\cite{Wigner1958} and Girko's~\cite{Girko1985} original convergence results. The criterion simply becomes
$$\sigma < |\theta|,$$
where $\sigma$ is the standard deviation of $A$.

\begin{prop}\label{prop:stability:randomWeak}
Under Assumption~\ref{assum:delta1/2} and if
$$\sigma < |\theta|,$$
any asymptotic feasible equilibrium of the model~\eqref{LVWM} is almost surely linearly stable.
\end{prop}
\begin{proof}
Consider $A/\sqrt{CS}$ so that $\sigma < |\theta|$. This is equivalent to say that $B=\lp\theta I + \frac{1}{\sqrt{CS}}A\rp$ possesses only eigenvalues with negative real parts. Define now the symmetric matrix $\tilde B = DB + B^TD$, where $D=\textrm{diag}\lp\frac{1}{\sqrt{2}}\rp$. Each non-diagonal entry of the matrix $\tilde B$ are independent of the others up to symmetry, and have mean zero and standard deviation $\sigma$. By the Wigner semicircle law and May's criterion, its eigenvalues are almost surely all negative when $S\to\infty$. The matrix $B$ is thus dissipative when $S\to\infty$, which implies that any feasible equilibrium of the model~\eqref{LVWM} is locally stable (see~\cite{Logofet1983Book}).
\end{proof}

Note that feasibility is required to conclude on the stability of the equilibrium. Indeed, we illustrate in Fig.~\ref{fig:delta1/2:May:JacobianVSInteractions} that a stable equilibrium in the sense of $J=(\theta I + A/\sqrt{CS})$ (as in~\cite{May1972}, i.e.~so that $\sqrt{SC}\sigma<|\theta|$), is stable in the sense of $J=\textrm{diag}(x^*)(\theta I + A/\sqrt{CS})$ only when it is feasible, which we showed is never the case in large systems (see Thm.~\ref{thm:feasibility:delta1/2}).

\begin{figure}[h!]
\centering
\includegraphics{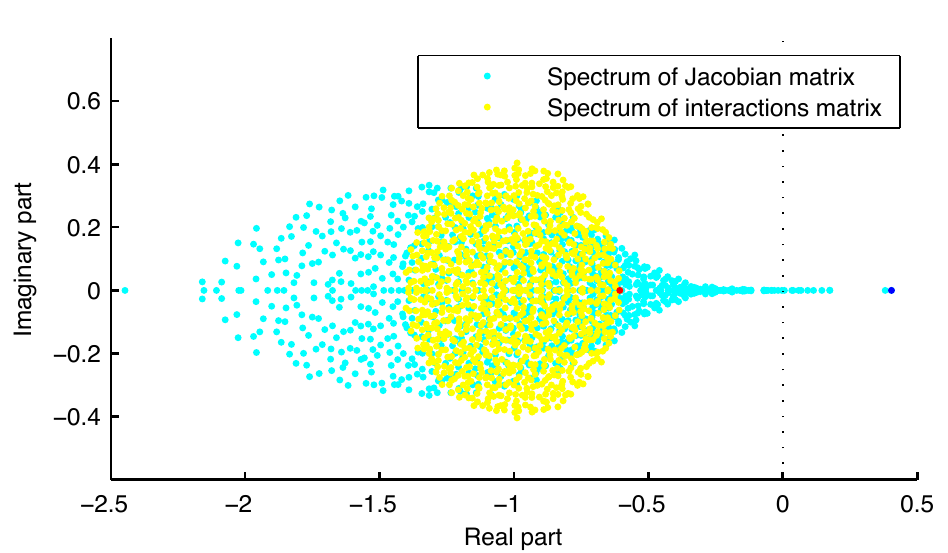}
\caption{Spectra of the interactions matrix and the Jacobian matrix for $\delta=0.5$ and $S=1000$. The spectrum of the interactions matrix $(\theta I + A/\sqrt{CS})$ is represented in yellow and the eigenvalue with largest real part is highlighted in red. With this renormalisation, the convergence result of Girko~\cite{Girko1985} is recovered and, since the parameters are $\sigma=0.4$, $C=1$ and $\theta = -1$, May's criterion holds. However, the equilibrium is very unlikely to be feasible with $S=1000$ (see Fig.~\ref{fig:main1}), so that the criterion of Prop.~\ref{prop:stability:randomWeak} is violated and the resulting equilibrium is not stable. This is illustrated by the spectrum of the Jacobian matrix $J(x^*)$ in cyan. The eigenvalue with largest real part is highlighted in blue.}\label{fig:delta1/2:May:JacobianVSInteractions}
\end{figure}

\subsubsection{Weak interactions}
In the case of weak interactions, May's criterion~\cite{May1972} and the criteria for competition and prey-predation established by Allesina and Tang~\cite{Allesina2012} are asymptotically trivially satisfied if the parameters are chosen so that the (deterministic) equilibrium is admissible. Indeed, the criterion becomes $\sigma/\sqrt{SC} < |\theta|$ and $\sigma/\sqrt{SC}\to 0$ for $S\to\infty$.

Concerning mutualistic networks, the criterion in~\cite{Allesina2012} is non-trivial. We show that this criterion still holds under the additional condition that $x^*$ is feasible: 
\begin{prop}
Under weak interactions (Assumptions~\ref{assum:delta1}) and for a mutualistic interaction matrix $A$, $x^*$ is locally stable when $x^*$ is feasible and when
\begin{equation}\label{eqn:stable:mut:weak}
\, \E(\vert a_{ij}\vert)<\vert\theta\vert.
\end{equation}
\end{prop}
\begin{proof}
the inequality~\eqref{eqn:stable:mut:weak} has been computed in~\cite{Allesina2012} using the Ger{\v s}gorin circles~\cite{Gershgorin1931}. With analogous arguments as in the proof of Prop~\ref{prop:stability:randomWeak}, we define the symmetric matrix $\tilde B = DB + B^TD$, where $D=\textrm{diag}\lp\frac{1}{\sqrt{2}}\rp$ and $B=\lp\theta I + \frac{1}{CS}A\rp$. Its largest eigenvalue has been computed in~\cite{Furedi1981} and is given by $\sqrt2\E(|a_{ij}|) + \sqrt2\theta$, which is negative when~\eqref{eqn:stable:mut:weak} holds. $B$ is thus dissipative and consequently the equilibrium locally stable~\cite{Logofet1983Book}.
\end{proof}

\subsection{Structured models}
Determining stability criteria in structured models is mathematically difficult, since the entries of the resulting random matrix $J(x^*)$ can become highly dependent. In this situation, a direct use of the classical results in~\cite{Geman1982,Tao2010,Sommers1988} is no longer possible. In Fig.~\ref{fig:delta1:structure:spectra}, we illustrate some spectra of $A$ and $J(x^*)$ for 50 networks of the three structured models to evidence their difference. We find that a link between feasibility and stability is very likely to exist, in a similar way as for unstructured networks. 

\begin{figure}[h!]
\centering
\includegraphics{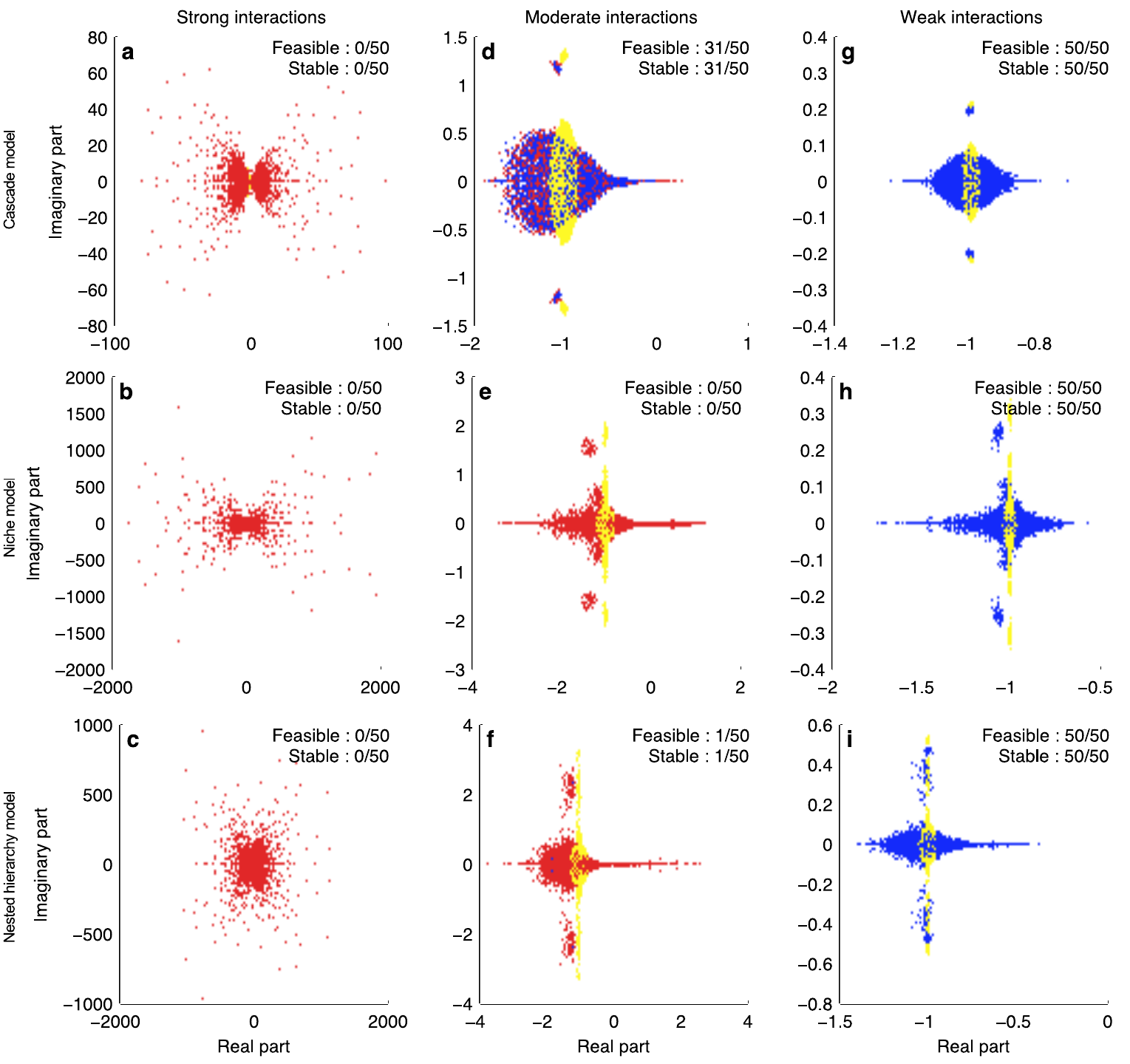}
\caption{Spectra of the Jacobian matrix for different values of $\delta$ in structured models. The eigenvalues of $J(x^*)$ are plotted in blue when $x^*$ is feasible, and in red otherwise; the eigenvalues of the interactions matrix $(\theta I + A/\sqrt{CS})$ is represented in yellow. (\textbf{a} - \textbf{c}) For strong interactions, the systems are degenerate for the three models, i.e. never feasible nor stable. (\textbf{d} - \textbf{i}) For moderate interactions and weak interactions, we observe the same kind of relationship between feasibility and stability as for unstructured models. For the cascade model (d, g), $x^*$ converges to a deterministic constant; for the niche and nested-hierarchy models (e, f, h, i), $x^*$ converges on a compact support, which can include negative values. It is thus necessary to tune the parameters to obtain feasibility; here, we chose $S=150$, $\sigma = 0.4$, $C=0.25$ and $\theta =-1$ in all cases.}\label{fig:delta1:structure:spectra}
\end{figure}

\begin{figure}[h!]
\centering
\includegraphics{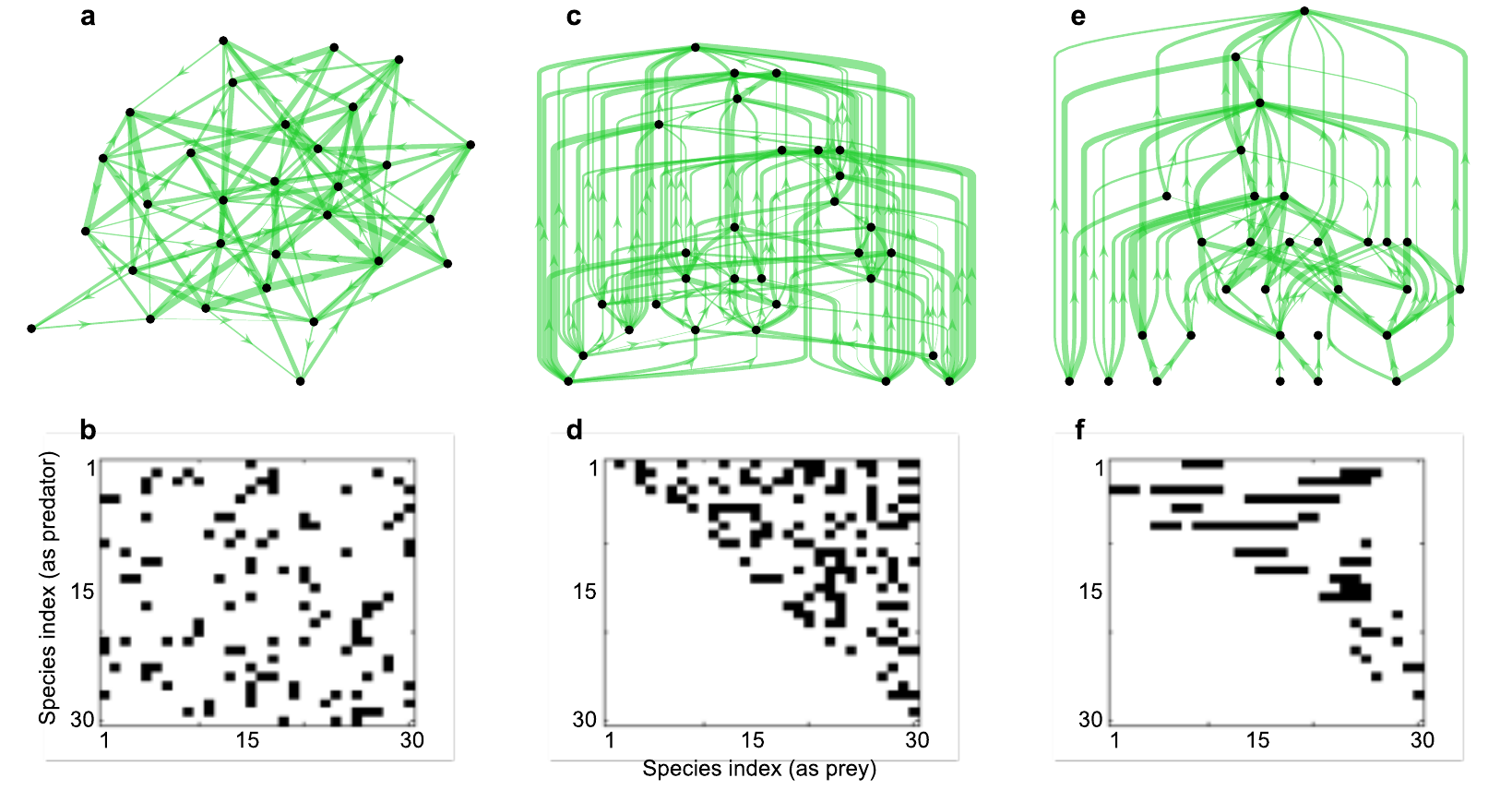}
\caption{Illustration of some prey-predator networks and adjacency matrices. (\textbf{a}-\textbf{b}) Unstructured model for predation. (\textbf{c}-\textbf{d}) The cascade model. (\textbf{e}-\textbf{f}) The niche model. The parameters are $C=0.25$ and $S=30$. The interaction strengths are sampled at random, according to the procedure described in section S.2.3. Thicker green arrows represent larger interaction strengths. Only positive interactions are represented.}\label{fig:networks}
\end{figure}

%%% Remove for submission
\bibliography{Feasibility_Arxiv}
\bibliographystyle{naturemag}

\section*{Acknowledgement}
We thank S.~Gray for useful comments on the manuscript.

%\clearpage

%\newpage

\end{document}